\documentclass[journal,onecolumn,11pt]{IEEEtran}
\usepackage[doublespacing,nodisplayskipstretch]{setspace}
\usepackage[utf8]{inputenc} 
\usepackage[T1]{fontenc}
\usepackage{url}
\usepackage{ifthen}
\usepackage{cite}
\usepackage[cmex10]{amsmath} 

\usepackage{eulervm}
\usepackage{mystyle}
\usepackage{xcolor}
\usepackage{graphicx}
\usepackage{physics}
\usepackage{amssymb}
\usepackage{comment}
\usepackage{amsthm}
\usepackage{subcaption}

\usepackage{thmtools}
\declaretheorem[name=Theorem, numberwithin=section]{theorem}

\newtheorem{remark}{Remark}
\newtheorem{example}{Example}
\newtheorem{definition}{Definition}

\newcommand{\mnew}[1]{\color{blue}#1\color{black}}

\newcommand{\neswarrow}{\mathrel{\text{$\nearrow$\llap{$\swarrow$}}}}
\newcommand{\nwsearrow}{\mathrel{\text{$\nwarrow$\llap{$\searrow$}}}}

\makeatother

\begin{document}
\title{One-at-a-Time Quantum Guessing: \\[-1ex]
Multipartite Entanglement Beyond MoE Games}

\author{Michael Schleppy
and Emina Soljanin
\thanks{M.~Schleppy and E.~Soljanin are with the Department of Electrical and Computer Engineering, Rutgers, the State University of New Jersey, Piscataway, NJ 08854, USA, e-mail: \{michael.schleppy, emina.soljanin\}@rutgers.edu).}
\thanks{This material is based upon work supported by the National Science Foundation Graduate
Research Fellowship Program under Grant No. 2233066. Any opinions,
findings, and conclusions or recommendations expressed in this material are those of the
author(s) and do not necessarily reflect the views of the National Science Foundation.}
}

\IEEEoverridecommandlockouts

\maketitle

\begin{abstract}
Multipartite entanglement remains a challenging and not fully understood aspect of quantum information. Monogamy-of-Entanglement (MoE) games have been highly effective for studying limitations on the usefulness of entanglement imposed by monogamy constraints. To better reveal the extent to which multipartite entanglement can be useful, we introduce a class of quantum guessing games, termed One-at-a-Time Guessing (OTG) games. In these games, quantum players individually guess the outcomes of random measurements performed by a referee on a pre-shared entangled state. Unlike MoE games, OTG games select players individually at random according to a specified probability distribution, thereby probing each player's correlation with the referee.
We show that, despite monogamy constraints, players sharing certain entangled states can moderately outperform those relying only on classical uncertainty. This advantage arises even in simple OTG games involving only Pauli measurements on qubits, where optimal entanglement increases the winning probability by at least 4\%. This contrasts with MoE games, where shared entanglement has been shown in several settings to provide only limited (if any) advantage over classical strategies.
We further establish a majorization property: the value of an OTG game respects the majorization ordering of the player-selection probability distribution. We also analyze in detail a two-player OTG game in which the referee measures one of the three Pauli observables on a qubit, and show that it is optimally played using a specific parameterized family of three-qubit $W$-like states. These results suggest that OTG games provide a useful framework for investigating the usefulness of multipartite entanglement in multiparty quantum correlations.

\end{abstract}

\begin{IEEEkeywords}
\noindent Guessing, Monogamy-of-Entanglement, Non-Local Games, Quantum Correlations, Quantum State Discrimination
\end{IEEEkeywords}

\newpage

\section{Introduction}

Consider the following scenario: Alice and Bob share a pair of photons which are maximally entangled, and Alice measures her photon using a single-photon detector. She measures the polarization of her photon in either one of two randomly selected axes, either  $\theta_1$ or $\theta_2$, which differ by a rotation of $45^{\circ}$. The outcome of her measurement will be polarization in either the vertical ($\updownarrow$) or horizontal $(\leftrightarrow)$ directions in case axis $\theta_1$ is chosen, or polarization in a diagonal direction, $(\neswarrow)$ or $(\nwsearrow)$, in case axis $\theta_2$ is chosen. Naturally, Bob would like to learn the outcome of Alice's measurement, and the entanglement of their photons permits him to do so. If, for instance, the entangled state of the photons was initially $\ket{\psi}=\frac{\ket{\uparrow\uparrow}+\ket{\downarrow\downarrow}}{\sqrt{2}}$, and Bob measures the polarization of his photon using the same axis as Alice, his outcome will be perfectly correlated with Alice's outcome. The perfect correlation of their outcomes holds regardless of which axis Alice originally measured her photon, and is a feature of maximal entanglement.

Now, suppose that Bob is not the only party who is interested in Alice's outcome, but also Charlie. The three parties share a triplet of entangled photons, and once again, Alice measures the polarization of her photon randomly in one of the axes $\theta_1,\theta_2$. How well can Bob and Charlie \textit{individually} learn Alice's measurement outcome if they cannot communicate with each other? The situation is not as simple as before, since it is impossible for both Bob and Charlie's photons to be maximally entangled with Alice's photon simultaneously. Consequently, they will have to accept some inevitable error in their guesses of Alice's measurement outcome. A natural follow-up goal, then, is to determine the minimum amount of error that Bob and Charlie individually observe, if they are selected to guess Alice's outcome at random. A further goal is to identify which tripartite states that enable Bob and Charlie to guess optimally, and to determine, if at all, how entanglement can benefit Bob and Charlie in spite of MoE. This is the premise of One-at-a-Time Guessing (OTG) games, the subject of this paper.

Guessing is a fundamental task with broad applications in information theory and cryptography. As a simple example, the security of user authentication in computer networks is influenced by the relative strength of user credentials; the harder it is to guess a password, the harder it is for an adversary to gain unauthorized access. In this view, the security of information is intrinsically linked to how difficult it is to guess. How then, do we measure the difficulty of guessing information? In practice, we model the secret information $X$ as a random variable, but we still must determine \textit{how} guesses are allowed to be made. This brings about two scenarios: iterative guessing and one-shot guessing. In the former, one seeks to determine the number of guesses required to determine $X$ exactly or within some accuracy, whilst receiving \textit{feedback} between guesses (e.g. `yes' and `no'). In the latter, one measures the difficulty of guessing by looking at the most probable outcome of the random variable $X$. Both settings relate to some \textit{entropic measure} of $X$, the R\'enyi entropy of order $\frac{1}{2}$ and the min-entropy, respectively. The latter scenario is the focus of this paper (for more information on iterative guessing, see \cite{Guess:Massey94,Guess:Arikan96,Guess:Arikan98,Guess:Malone04,Guess:Sundaresan06,Guess:Hanawal10,Guess:Christiansen12,Guess:Christiansen15,Guess:Hanson21,Guess:Rioul22}).

The one-shot guessing scenario derives its name from \textit{one-shot information theory}, in which broadly speaking, resources (e.g. channels, quantum states, measurements) are limited to one use. A concrete example of such an information-theoretic task is one-shot hypothesis testing. In (parametric) hypothesis testing, the goal is to decide between a selection of hypothesis $x \in \mcX$ based upon some observation $y \in \mcY$, where the distribution of $y$ is determined by the parameter $x$. In the Bayesian setting, we assume further that the hypotheses $x \in \mcX$ are distributed according to some prior distribution $p_X(x)$. The performance of a Bayesian hypothesis test under some decision rule $\delta:\mcY \to \mcX$ is measured by the probability of misidentification, i.e. $\Pr[X \neq \delta(Y)]$ (see \cite{HT:Moulin18} for a general overview of hypothesis testing). In the quantum setting, hypothesis tests frequently take the form of \textit{quantum state discrimination} tasks, where the goal is to distinguish between ensembles of quantum states with maximal probability using some measurement. The origins of quantum state discrimination are due to Helstrom \cite{HT:Helstrom69} (see \cite{HT:Barnett09,HT:Bae15} for an overview of quantum state discrimination, also \cite[Chapter 3.1]{HT:Watrous18}). The optimal ability to guess a random variable $X$, perhaps with some random observation $Y$, is innately tied to the optimal accuracy of the corresponding one-shot hypothesis test.

One-shot guessing is also particularly relevant in scenarios where communication is limited (or outright impossible), due to communication cost, latency, bandwidth, privacy, etc. In this setting, the principal objects of study are the \textit{correlations} between distant parties. When we discuss correlations between two distant parties, we refer to the joint distribution $p(x,y|a,b)$ between two random variables $X$ and $Y$, conditioned on some local observations $A$ and $B$ respectively. The chief task in the study of correlation sets is to determine which correlations $p(x,y|a,b)$ are realizable under various non-communicating resource constraints (i.e. shared randomness, entanglement, commuting observables, non-signaling). A famous result due to Bell asserts that there exist realizable correlations in quantum mechanics, that are not achievable using classical resources (i.e. shared randomness) \cite{Correlation:Bell64}.

The power of quantum correlations over classical correlations is enabled by the properties of \textit{entanglement} and \textit{steering} in quantum systems. Informally, a system of distant quantum particles is entangled when its state cannot be determined by local descriptions of each particle, and steering is the ability for measurement outcomes on one particle to influence the state of another particle via entanglement. Crucially, entanglement is most pronounced and measurable on bipartite systems. This is encapsulated by the fact that the outcome $X$ of any noiseless measurement\footnote{By noiseless measurement, we mean a quantum measurement using a projection-valued measure.} on one subsystem can be perfectly replicated (or guessed) on another subsystem, so long as they share a maximally entangled state. There is no way for a third party to replicate $X$ sufficiently well without weakening the correlations between the first two parties, a fact upon which the security of quantum key distribution lies. This notion is further strengthened by the principle of \textit{Monogamy-of-Entanglement}, which informally states that if the states of two systems $A$ and $B$ are maximally entangled, then the third system $C$ must be completely unentangled from both $A$ and $B$. This certainly does not preclude the existence of states which contain entanglement between three or more subsystems (notable examples include the GHZ state and W state), but it does prevent systems of three or more parties from sharing a single state that, when restricted to two subsystems, functions exactly like a maximally entangled state.

\subsection{Non-Local Games and Monogamy-of-Entanglement Games}
A powerful tool for analyzing the abilities of different correlation types under resource constraints is the theory of non-local games (and variants thereof). A non-local game can be thought of as an experiment in which a referee sends random `questions' to a group of non-communicating players, who perform local actions (measurements, computations, etc.) to produce `answers' that are submitted to and evaluated by the referee. The players are awarded a score depending upon the questions dispensed and answers received. The performance of the players is measured by their expected score.

Non-local games have been studied extensively across areas in physics and computer science. The first notion of non-local games were experiments proposed to test (and eventually disprove) local hidden-variable theories, including experiments regarding the CHSH inequality \cite{NLG:Clauser69}. Later, the non-local game framework gained traction in complexity theory, in which non-local games are central in a number of complexity class models such as MIP and MIP* (see e.g. \cite{NLG:Vidick16}). Finally, non-local games serve as an invaluable tool for device-independent quantum cryptography, where games sere as tools used to rule out the presence of eavesdroppers, and certify the security of a generated key from a key distribution protocol using observed statistics. For a comprehensive overview on non-local games in device independent quantum cryptography, see \cite{DIQKD:Zapatero23,DIQKD:Primaatmaja23}. For an overview of quantum correlations in nonlocal games, refer to \cite{NLG:Palazuelos16}.

In non-local games, the role of the referee is purely classical: sending/receiving classical information and evaluating a score function. Recent works have investigated non-local games in which the referee has quantum capabilities, with one such approach being extended non-local games \cite{MOE:Johnston16,MOE:Russo17}. Instead of evaluating a score function, the referee measures an observable on part of a shared state based upon the configuration of questions/answers. An important subclass of extended non-local games are \textit{monogamy-of-entanglement games} \cite{MoE:Tomamichel13}, in which the referee performs a random measurement, whose outcome is to be simultaneously guessed by non-communicating players. Monogamy-of-entanglement (MoE) games naturally model a plethora of scenarios in quantum communications (e.g. the referee is a trusted party in QKD), and have been recently adopted as a valuable tool in unclonable cryptography \cite{MOE:Broadbent19,MOE:Coladangelo21,MOE:Ananth23,MOE:Botteron24,MOE:Poremba24}. Notably, MoE games are games of equal information; all questions given to players are identical (compare with non-local games, where the complexity is generally born out of giving players independently random questions). 

\subsection{Our Contributions and Paper Organization}
In this paper, we introduce a new variant of extended non-local games which emphasizes \textit{individual} players' ability to guess the referee's measurement outcome. We refer to these games as \textit{One-at-a-Time Quantum Guessing games}, or OTG games for short. The difference between MoE games and our proposed OTG games lies in who makes guesses: in MoE games, all players must \textit{simultaneously} guess the measurement outcome of the referee's measurement, whereas in OTG games, only one player must make a guess at a given time. The individual player is selected from a known probability distribution, and the players' goal is to choose a shared state and measurements that optimizes the average ability of players to guess the referee's measurement outcome.

The motivation for OTG games is multi-fold. Firstly, by restricting the focus to players' individual guessing ability, we de-emphasize the somewhat strict requirement that all players simultaneously agree on answers. In doing so, we shift the focus onto bipartite correlations between the referee and individual players, and consequently, identify states which optimally balance entanglement from the referee to each player. For instance, OTG games may help to reason about entanglement resource allocation, particularly in scenarios where the referee must decide in real-time to perform a protocol (i.e. QKD, quantum teleportation) with a player selected at random, using pre-shared entanglement. Secondly, OTG games are a new tool to investigate the role of \textit{nonlocality} in hypothesis testing. When observations are fundamentally mediated to observers through shared states and entanglement, the effects of nonlocality become apparent (i.e. how much of the state does an observer posses, and how does that affect their ability to distinguish hypotheses). OTG games, along with other related games, both propose new foundational questions and offer new tools for understanding the role of nonlocality in hypothesis testing. 


Our paper is organized as follows: in Section~\ref{sec:prelim}, we review notation and preliminaries relevant to OTG games and our results. In Section~\ref{sec:OTG_Games}, we formally introduce OTG games and related notions, compare them against MoE games, and state a variety of results regarding the optimal winning probability of each. Crucially, we state a majorization result relating the optimal guessing probability of an OTG game to its player selection probability. In Section~\ref{sec:PauliOTG}, we investigate a specific OTG game where the referee uses measurements in Pauli observables on a single qubit. We classify the optimal guessing probability as a function of the player selection probability, and identify a family of $W$-like tripartite states which help the players guess optimally in the OTG game. The optimality of the states and measurements are certified using semidefinite programs from the NPA hierarchy. Proofs and auxiliary material are relegated to the Appendix.

\begin{figure}
    \centering
    \includegraphics[width=0.8\linewidth]{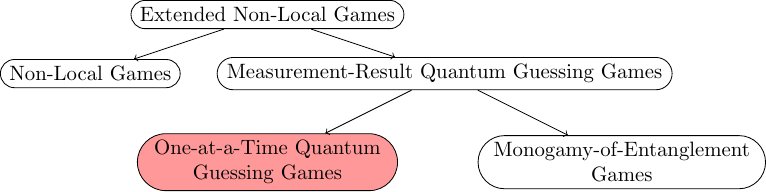}
    \caption{A hierarchy of related quantum games. Arrows denote inclusion (e.g. measurement-result quantum guessing games are a subclass of extended non-local games). We propose the variant of one-at-a-time quantum guessing games.}
    \label{fig:gameHierarchy}
\end{figure}


\section{Preliminaries: Notation and Background Material}\label{sec:prelim}

Throughout this paper, $\mcH$ will denote a finite-dimensional complex Hilbert space. Subscripts will be used to distinguish between different systems, e.g., $\mcH_A$ for Alice, $\mcH_B$ for Bob, etc. When we discuss composite systems, we use the shorthand notation $\mcH_{A_1A_2...A_n} = \mcH_{A_1} \otimes \mcH_{A_2} \otimes \ldots \mcH_{A_n}$. A particular component $A_i$ of a composite system on $\mcH_{A_1...A_n}$ will be referred to as either a subsystem or a register, interchangeably. In case a particular register $A_i$ holds only classical information (as opposed to quantum information), we may write $C_{A_i}$ instead of $\mcH_{A_i}$. We will denote the set of positive operators on $\mcH$ by $\mcP(\mcH)$. The identity operator on $\mcH$ will be denoted by $\bbI_{\mcH}$, or simply $\bbI$ if the underlying Hilbert space is clear. The adjoint of an operator $X$ on $\mcH$ will be denoted as $X^\dagger$. Familiarity with Dirac notation is assumed.

A density operator $\rho:\mcH \to \mcH$ on a Hilbert state $\mcH$ is defined to be a positive-semidefinite operator on $\mcH$ with unit trace. We will refer to density operators as states, and we denote the set of such states on a Hilbert space as $\mcD(\mcH)$. A pure state is a state with rank one, and has the form $\ket{\psi}\bra{\psi}$ for some normalized vector $\ket{\psi} \in \mcH$. By the spectral theorem, any state $\rho \in \mcD(\mcH)$ is a convex combination of pure states, namely 
\[\rho = \sum_{j \in J} p_j \ket{\psi_j}\bra{\psi_j}\]
where $p_j \geq 0, \sum_j p_j=1$, and $\ket{\psi_j} \in \mcH$ are normalized vectors, and $|J|\leq \dim(\mcH)$. Occasionally, we may want to refer to classical-quantum (cq) states, which are those states which are defined on a composite system $C \otimes \mcH$ containing a classical register and a quantum register. A cq state has the form
\[\rho_{CQ} = \sum_{j\in J} p_j \ket{j}\bra{j} \otimes \rho_j\]
where $\{\ket{j}\}_{j \in J}$ is a finite, fixed orthonormal basis of $C$, $p_j$ is a probability distribution on $J$, and $\rho_j \in \mcD(\mcH)$. The set of cq states on $C \otimes \mcH$ will be denoted by $\mcD(C \otimes \mcH)$, and the set of positive operators with a classical register $C$ and a quantum register $Q$ on a Hilbert space $\mcH$ will be denoted by $\mcP(C\otimes \mcH)$. A state on a composite system may be written with subscripts, e.g. a state on $\mcH_{A_1A_2...A_n}$ may be denoted by $\rho_{A_1A_2...A_n}$. We denote the partial trace of a state $\rho_{A_1A_2...A_n}$ on a composite system by omitting subscripts. For instance, if $I=\{i_1,i_2,\ldots,i_k\} \subseteq \{1,2,\ldots,n\}$, then $\rho_{A_{i_1}A_{i_2}...A_{i_k}}$ is the reduced state of $\rho_{A_1A_2...A_n}$ on the Hilbert space $\mcH_{A_{i_1}A_{i_2}...A_{i_k}}$, and is obtained by tracing out the subsystems $A_i$ with $i \in I^c$.


A positive operator-valued measure (POVM) on $\mcH$ with outcomes in a finite set $\mcX$ is a mapping $\mcM:\mcX \to \mcP(\mcH)$ such that $\sum_{x \in \mcX} \mcM_x = \bbI$. In case all $\mcM_x$'s are orthogonal projections (i.e. $\mcM_x^2 = \mcM_x^\dagger = \mcM_x$), we say that $\mcM$ is a projection-valued measure (PVM). Furthermore, in case all projections are rank one (i.e. $\mcM_x=\ket{\psi_x}\bra{\psi_x}$), we say that $\mcM$ is a basis measurement in the basis $\{\ket{\psi_x}\}_{x \in \mcX}$. For each state $\rho \in \mcD(\mcH)$, a POVM induces a set of measurement statistics $p=(p_x:x \in \mcX) \in \bbR^{|\mcX|}$, where $p_x=\Tr\left[\mcM_x \rho\right]$. It follows immediately from the definition of a POVM that $p$ is a probability distribution on $\mcX$.

A \textit{qubit} is represented by a state on a two-dimensional Hilbert space $\mcH=\bbC^2$. The space of Hermitian linear operators on $\bbC^2$ is a real Hilbert space with inner product $\langle X,Y\rangle = \Tr[XY]$, and an orthogonal basis given by the Pauli operators.
\[\bbI = \begin{pmatrix} 1 & 0\\0 & 1 \end{pmatrix} ~~~ 
     X = \begin{pmatrix} 0 & 1\\1 & 0 \end{pmatrix} ~~~
     Y = \begin{pmatrix} 0 & -i\\i & 0 \end{pmatrix} ~~~
     Z = \begin{pmatrix} 1 & 0\\0 & -1 \end{pmatrix}\]

Consequently, states and measurements on qubits can be expressed via Pauli operators. Recall that any Hermitian matrix $H$ on $\mcH$ has an eigendecomposition $H=\sum_{i=1}^r\lambda_iP_i$, where $\{\lambda_i\}_{i=1}^r$ are distinct real numbers, and $\{P_i\}_{i=1}^r$ are orthogonal projections summing to identity. When we refer to a measurement $\mcM$ with observable $H$, we mean that $\mcM$ is a projection-valued measure with PVM elements given by the orthogonal projectors $P_i$ of $H=\sum_{i=1}^r\lambda_iP_i$. For qubits, we identify three distinguished binary-outcome measurements, which come from the Pauli observables.

\begin{equation*}
    \begin{array}{ccc}
    X=\ket{+}\bra{+}-\ket{-}\bra{-} ~~~ & Y=\ket{+i}\bra{+i} - \ket{-i}\bra{-i} ~~~ & Z = \ket{0}\bra{0} - \ket{1}\bra{1}\\ 
    \mcM_0 = \ket{+}\bra{+} & \mcM_0 = \ket{+i}\bra{+i} & \mcM_0 = \ket{0}\bra{0}\\
    \mcM_1 = \ket{-}\bra{-} & \mcM_1 = \ket{-i}\bra{-i} & \mcM_1 = \ket{1}\bra{1}\\
\end{array}
\end{equation*}

where $\ket{\pm} = \frac{\ket{0}\pm\ket{1}}{\sqrt{2}}$ and $\ket{\pm i}=\frac{\ket{0}\pm i\ket{1}}{\sqrt{2}}$ are the Hadamard basis and $Y$-eigenbasis, respectively. The basis $\{\ket{0},\ket{1}\}$ is called the computational basis.

Given a probability mass function (PMF) $p$ on a finite set $\mcX$, we write $p^{\downarrow} \in \bbR^{|\mcX|}$ to be the vector whose coefficients are non-increasing and are some permutation of the coefficients of $p$, namely
\[p^\downarrow = (p^\downarrow_1,p^\downarrow_2,\ldots,p^\downarrow_{|\mcX|}), \text{ such that } p^\downarrow_1 \geq p^\downarrow_2 \geq \ldots \geq p^\downarrow_{|\mcX|} \text{ and } p^\downarrow_x = p_{\sigma(x)} \text{ for some } \sigma \in S_\mcX\]
where $S_\mcX$ denotes the set of permutations on $\mcX$. Given two probability distributions $p,q$ on $\mcX$, we say that $p$ \textit{majorizes} $q$ if for every $k \in \{1,2,\ldots |\mcX|\}$, we have
\[\sum_{x=1}^k p^\downarrow_x \geq \sum_{x=1}^k q^\downarrow_x.\]
We write $p \succeq q$ to denote that $p$ majorizes $q$, and $p \preceq q$ to denote that $q$ majorizes $p$. Majorization is a partial order on the probability simplex on $\mcX$, with a minimum element $p_{\text{Unif}} = \left(\frac{1}{|\mcX|},\ldots, \frac{1}{|\mcX|}\right)$ and maximal elements $\delta_x$, where $\delta_x$ is the indicator function on $\{x\}$ (see \cite{Majorization:Marshall79} for a reference on majorization).

The set of integers $\{1,2,\ldots,m\}$ will often be denoted by $[m]$. Given two integer-valued vectors $\underline{d},\underline{d}' \in \bbN^m$, we will often make use of the shorthand notation $\underline{d}\preceq \underline{d}'$ to mean that $d_i \leq d_i'$ for all $i \in [m]$. We denote the set of permutations on $[m]$ as $S_m$. The notation $\delta_{xy}$ will denote the Kronecker-delta function, i.e. $\delta_{xy}=1$ if $x=y$, else $\delta_{xy}=0$. 


\section{One-at-a-Time Guessing Games}\label{sec:OTG_Games}
We now proceed to the main subject of the paper, \textit{One-at-a-Time Guessing Games}. To motivate the OTG Game setting, we begin with an illustrative example, inspired by the scenario of the introduction.
\begin{example}\label{ex:OTG_with_GHZ}
    Alice, Bob, and Charlie each posses one out of a triplet of qubits, whose composite state $\ket{\psi}_{ABC}=\frac{1}{\sqrt{2}}\left(\ket{000}+\ket{111}\right) \in \mcH_{ABC} = \bbC^2 \otimes \bbC^2 \otimes \bbC^2$ is the GHZ state, and is known to everyone. Alice seeks to determine how well Bob or Charlie can guess the outcome of a measurement on her qubit. She selects one measurement $\mcM_b$ at random:
    \[\mcM_{0}=\bigl\{\ket{0}\bra{0},\ket{1}\bra{1}\bigr\} ~~~~ \text{or} ~~~~ \mcM_{1}=\bigl\{\ket{+}\bra{+},\ket{-}\bra{-}\bigr\}\]
    Alice records the measurement choice $b \in \{0,1\}$ as well as the measurement outcome $x \in \{0,1\}$. Now, she randomly designates one of Bob or Charlie as a guesser, whose goal is to determine $x$. To this end, Alice tells the designated guesser which measurement choice $b \in \{0,1\}$ she selected.

    Suppose Bob is the designated guesser. Bob will choose a measurement $\mcM'_b$ based upon the measurement choice $b$. If Alice's measurement choice is $b=0$, Bob computes the post-measurement states $\rho_x$ on his qubit for each $x \in \{0,1\}$, which are 
    \[\rho_x=\ket{x}\bra{x} \text{ with probability } \frac{1}{2}\]
    Consequently, a measurement in the computational basis is sufficient for Bob to guess $x$ perfectly. Now, suppose Alice's measurement choice is $b=1$. The post-measurement states $\rho_x$ on Bob's qubit are now
    \[\rho_x = \frac{1}{2}\bbI \text{ with probability } \frac{1}{2}\]
    In this case, no measurement will help Bob determine $x$ with probability greater than $\frac{1}{2}$. The situation is identical if Charlie were instead selected as the designated guesser.

\end{example}

The above example illustrates how measurements on subsystems of an entangled system can induce \textit{steering} on the other subsystems. When Alice measures her qubit in the computational basis, the post-measurement states on Bob and Charlie's qubits, conditioned on Alice's measurement outcome, are either $\ket{00}$ or $\ket{11}$, both of which can be distinguished by Bob or Charlie individually. When Alice measures in the Hadamard basis, the post-measurement states on Bob and Charlie's qubits, conditioned on Alice's measurement outcome, are either $\frac{1}{\sqrt{2}}\left(\ket{00}+\ket{11}\right)$ or $\frac{1}{\sqrt{2}}\left(\ket{00}-\ket{11}\right)$, neither of which are distinguishable by Bob or Charlie individually.

In general, extended non-local games involve two types of participants: a referee ($R$) and a set of players ($P_i$).

\begin{itemize}
    \item \textbf{Referee:} A referee ($R$) is an \textit{honest} agent who interacts with others in a \textit{characterized} manner. 
    \item \textbf{Player:} A player ($P_i$) is an agent who can strategize arbitrarily with other players.
\end{itemize}

The exact formulation of extended non-local games can be found in \cite{MOE:Russo17,MOE:Johnston16}. Here, we define a subclass of extended non-local games, called measurement-result quantumm guessing games, which generalize and provide a framework for comparing both the classes of MoE games and OTG games
\subsection{Measurement-Result Quantum Guessing Games}

In measurement-result quantum guessing games, the referee holds a register of fixed dimension, and performs random measurements $\mcM^\theta$ on their register according to a known distribution $p_{\Theta}$. The players aim to determine the referee's measurement outcome, based upon knowledge of the measurement choice $\theta$. Such games can typically be divided up into several phases:

\begin{itemize}
    \item \textbf{Preparation Phase:} The players $P_1,P_2,\ldots,P_m$ prepare a multi-partite state $\rho_{RP_1...P_m}$ to be shared amongst the referee and all players.
    \item \textbf{Referee Measurement Phase:} The referee measures their subsystem in some measurement $\mcM^\theta$ according to distribution $p_\Theta$, and receives measurement outcome $x$. The measurement choice $\theta$ is announced to all players.
    \item \textbf{Player Measurement Phase:} The players perform local measurements on their subsystems to obtain guesses $x_i$ for the referee's measurement outcome $x$, without communicating.
    \item \textbf{Verdict Phase:} The referee evaluates the guesses from a subset $I\subseteq [n]$, determining if $x_i=x$ for all $i \in I$. If all such players match the referee's measurement outcome, they win. Otherwise, they lose.
\end{itemize}

The novelty in OTG games lies specifically in the verdict phase. There are naturally two extreme cases for the evaluation of guesses: evaluating the entire set of players' guesses ($I=[m]$), or evaluating an individual player's guess randomly ($I=\{k\}$ for some random $k \in [m]$). MoE games cover the former case, whereas OTG games cover the latter case.

We now provide a formal definition for the components of a measuement-based quantum guessing game, outlined in the phases above.

\begin{definition}\label{def:ref_config}
    A \underline{Referee Configuration} is defined by the tuple of objects $ \mcR = (\mcH_R,\Theta,p_\Theta,\mcX,\{\mcM^\theta\}_{\theta \in \Theta})$, where

    \begin{itemize}
        \item $\mcH_R$ denotes the referee's Hilbert space (which is assumed to be finite dimensional).
        \item $\Theta$ denotes a (finite) parameter set. This set is also referred to as a set of questions.
        \item $p_\Theta$ is a probability distribution on $\Theta$.
        \item $\mcX$ denotes a (finite) set of outcomes. This set is also referred to as a set of answers.
        \item For each $\theta \in \Theta$, $\mcM^\theta$ is a POVM measurement on $\mcH_R$, with POVM elements $\mcM^\theta_x$ and outcomes in $\mcX$.
    \end{itemize}
\end{definition}

The Referee Configuration specifies exactly how the referee behaves during the referee measurement phase of the game: a measurement $\mcM^\theta$ is selected with probability $p_\Theta$, and the referee measures the register $\mcH_R$ with POVM $\mcM^\theta$ to get an outcome $x \in \mcX$. We therefore say that the referee's behavior is \textit{characterized}. In contrast, players can strategize before the game begins and agree to perform arbitrary measurements, in order to guess the referee's measurement outcome correctly. But first, the players must know when and how they are selected to produce guesses. This is where the evaluation procedure for the game becomes relevant.

\begin{definition}\label{def:eval_proc}
    An \underline{Evaluation Procedure} for $m$ players is a tuple of objects $\mcE = (\mcI,q_I)$ such that

    \begin{itemize}
        \item $\mcI$ is a collection of subsets of $[m]$, called the player query set, denoting sets of players that can be asked to simultaneously guess.
        \item $q_I$ is a probability mass function on $\mcI$, called the selection function.
    \end{itemize}
\end{definition}

The evaluation procedure serves to determine which subsets of players must simultaneously guess the referee's measurement outcome, along with their corresponding frequency. Specifically, when a subset of players $A = \{i_1,i_2,\ldots,i_k\} \in \mcI$ is chosen with probability $q_I(A)$, each player $i_j \in A$ performs a local measurement on their subsystem to obtain a guess $x_j \in \mcX$ for the referee's measurement outcome $x$. The guesses $x_1,x_2,\ldots,x_k$ are submitted to the referee to be evaluated, where the players win if and only if $x=x_1=x_2=\ldots=x_k$. The role of the distribution $q_I$ is to test correlations between different subsets of players and the referee, each with varying frequency. The larger the subset $A \in \mcI$ is, the more restrictive the condition on winning is, since all players need to guess the referee's outcome correctly. As will be evident later on, an asymmetrical distribution $q_I$ between the players will generally favor, for optimal guessing, the use of states that prioritize entanglement between the referee and players who are selected with high frequency.

Both a referee configuartion $\mcR$ and an evaluation procedure $\mcE$ for $m$ players are suffcient to describe a measurement-result quantum guessing game, or MRQG game for short.

\begin{definition}
    A \underline{Measurement-Result Quantum Guessing Game} is a tuple $\mcG=(\mcR,\mcE)$ consisting of a referee configuration $\mcR$ and an evaluation procedure $\mcE$.
\end{definition}


\subsection{Strategies for Measurement-Result Quantum Guessing Games}


As illustrated in Example~\ref{ex:OTG_with_GHZ}, players interact with the referee via one-way communication and a shared state $\rho \in \mcD(\mcH_{R} \otimes \mcH_P)$, where $\mcH_P$ denotes the joint Hilbert space of the players. The goal of the players in MRQG games is to use the parameter $\theta \in \Theta$ to guess the referee's measurement outcome $x \in \mcX$. In particular, the referee's measurement $\mcM^\theta$ on register $R$ of the shared state $\rho \in \mcD(\mcH_R \otimes \mcH_P)$ yields an ensemble of post-measurement states $(\rho_x^\theta)_{x \in \mcX}$ on the players' shared Hilbert space, between which the players aim to discriminate. We formally define a strategy for a MRQG game below: 

\begin{definition}\label{def:MoE_Strategy}
    A strategy for a Measurement-Result Quantum Guessing Game is defined by a tuple\\ $\mcS = (\mcH_{P_1...P_M},\rho_{RP_1...P_m}, \{\mcP^\theta_i\}_{\theta \in \Theta, i \in [m]})$, where:
    \begin{itemize}
        \item $\mcH_{P_1...P_m} = \mcH_{P_1} \otimes \mcH_{P_2} \otimes \ldots \otimes \mcH_{P_m}$ is the joint Hilbert space of the players. The i-th player holds the i-th register, corresponding to the Hilbert space $\mcH_{P_i}$.
        \item $\rho_{RP_1...P_m} \in \mcD(\mcH_R \otimes \mcH_{P_1...P_M})$ is a shared state between the referee and the players.
        \item For each $\theta \in \Theta, I \in \mcI$ and $i \in I$, $\mcP^{\theta,I}_i$ is a POVM measurement for the i-th player on $\mcH_{P_i}$, with POVM elements $\mcP^{\theta,I}_{i,x}$ and outcomes in $x\in \mcX$.
    \end{itemize}
\end{definition}


The probability of winning the MRQG game $\mcG$ using a strategy $\mcS$ will be denoted by $p_{\mcG}(\mcS)$, and takes the following form.
\begin{equation*}\label{eq:MoE_win_prob}
    p_{\mcG}(\mcS) = \sum_{\theta \in \Theta} \sum_{I \in \mcI} p_{\Theta}(\theta) q_I(I) \sum_{x \in \mcX} \Tr \left[ \left( \mcM^\theta_x \otimes \mcP^{\theta,I}_{i_1,x} \otimes \mcP^{\theta,I}_{i_2,x} \otimes \ldots \otimes \mcP^{\theta,I}_{i_k,x}\right) \rho_{RP_{i_1}...P_{i_k}} \right]
\end{equation*}

The players aim to select a strategy that maximizes the probability of winning. We define the \textit{value} of a MRQG game $\omega(\mcG)$ to be the supremum of all winning probabilities across all strategies:
\begin{equation*}\label{eq:MoE_value}
    \omega(\mcG) = \sup \{p_\mcG(\mcS):\mcS \text{ is a strategy for $\mcG$}\}
\end{equation*}

Occasionally, it is of interest to determine the optimal winning probability when the dimensions of the players' subsystems are bounded in dimension. Letting $\underline{d}=(d_1,\ldots,d_m) \in \bbN^m$, we define a $\underline{d}$-bounded
strategy for an MRQG game $\mcG$ to be a strategy in which the i-th player is bounded to use a Hilbert space of dimension at most $d_i$. The $\underline{d}$-bounded value $\omega(\mcG,\underline{d})$ of a MRQG game $\mcG$ is defined as the largest achievable winning probability under a $\underline{d}$-bounded strategy.
\begin{equation*}\label{eq:MoE_constrained_value}
    \omega(\mcG,\underline{d}) = \sup \{p_\mcG(\mcS):\mcS \text{ is a strategy for $\mcG$, $\dim(\mcH_{P_i}) \leq d_i$}\}
\end{equation*}

When $d_i=D$ is the same for each $i \in [m]$, we simply say that a $\underline{d}$-bounded strategy is $D$-bounded, and write $\omega(\mcG,D)$ instead of $\omega(\mcG,\underline{d})$. A $D$-bounded strategy with $D=1$ will be referred to as a \textit{classical} strategy, whereas a $\underline{d}$-bounded strategy with $d_i>1$ for all $i \in [m]$ will be referred to as a \textit{quantum} strategy. If $d_i > 1$ for some but not all $i \in [m]$, the strategy will be referred to as a \textit{classical-quantum} strategy. The choice of names for the classes of strategies is appropriate, since a player with a Hilbert space of dimension $d_i=1$ must guess an outcome based on the classical uncertainty of the referee's measurement, without assistance from entanglement.

One may also consider randomized strategies for MRQG games. A randomized strategy $\mcS^{\text{rand}}$ is a probabilistic combination of strategies $\mcS_1,\mcS_2,\ldots,\mcS_k$, where strategy $\mcS_j$ is chosen with probability $r_j$. The winning probability of a randomized strategy is given by $p_{\mcG}(\mcS^{\text{rand}}) = \sum_j r_j p_{\mcG}(\mcS_j)$. Formally, a randomized strategy can be implemented with shared state $\rho=\sum_j r_j \ket{j}\bra{j}^{\otimes m} \otimes \rho_j \in \mcD(C_{J^m} \otimes \mcH_{RP_1...P_m})$ and player $P_i$'s measurements given by the POVM elements $\mcP^{\theta,I}_{i,x} = \sum_j \ket{j}\bra{j} \otimes (\mcP_j)^{\theta,I}_{i,x}$ in $\mcP(C_J \otimes \mcH_{P_i})$. The states $\rho_j$ and the POVM elements $(\mcP_j)^{\theta,I}_{i,x}$ are taken from strategy $\mcS_j$. The classical register $C_J$ contains the randomness seed $j \in J$, and a copy of $C_J$ is provided freely to each player. We abide by the convention that the classical register $C_J$ does \underline{not} increase the dimensions of the individual players' Hilbert spaces in $\underline{d}$-bounded strategies. This choice is made for two reasons: 1) randomized strategies do not outperform deterministic strategies, so $\omega(\mcG)$ and $\omega(\mcG,\underline{d})$ are unaffected if we include randomized strategies into their definitions, and 2) randomized strategies can be simulated before the game begins, by having the players agree to perform strategy $\mcS_j$ with probability $r_j$. 

Finally, a MRQG game is said to be \textit{symmetric} if its evaluation procedure $\mcE$ is invariant under permutations of the players' indices. That is, for any $I = \{i_1,i_2,\ldots,i_k\} \in \mcI$ and permutation $\sigma \in S_m$, we have $\sigma(I)=\{\sigma(i_1),\sigma(i_2),\ldots,\sigma(i_k)\} \in \mcI$, and furthermore $q_I(I) = q_I(\sigma(I))$. In the context of $D$-bounded strategies, we consider a strategy $\mcS$ to be \textit{symmetric} if $\mcP^{\theta,I}_{i,x} = \mcP^{\theta,\sigma(I)}_{\sigma(i),x}$ for all $i \in [m]$ and permutations $\sigma \in S_m$, and the shared state $\rho_{RP_1...P_m}$ is invariant under any permutation of the players' subsystems (i.e. $\rho_{RP_{\sigma(1)}...P_{\sigma(m)}} = \rho_{RP_1...P_m}$, where $\sigma \in S_m$). 

\subsection{Monogamy-of-Entanglement Games and One-at-a-Time Guessing Games}\label{sec:OTG}


Here, we discuss a few specific points about both MoE and OTG games. As a quick note, we drop the index $I$ from the players' POVM elemens $\mcP^{\theta,I}_{i,x}$, as it becomes redundant in both settings.

A MoE game is always specified by a fixed number of players, due to the definition of the evaluation procedure. It is naturally of interest, however, to determine how the value of the game changes as more players are added. To this end, for a fixed referee configuration $\mcR$, we consider the family of MoE games $\mcG_m=(\mcR,\mcE_m)$, where $\mcE_m=(\mcI_m,q_{I,m})$ is the $m$-player evaluation procedure for a monogamy-of-entanglement game, i.e. $\mcI_m=[m]$ and $q_{I,m}([m])=1$. The winning probability of a strategy $\mcS$ for the $m$-player MoE game $\mcG_m$ is
\begin{equation*}
    p_{\mcG_m}(\mcS) = \sum_{\theta \in \Theta}p_{\Theta}(\theta)\sum_{x \in \mcX} \Tr\left[\left(\mcM^\theta_x \otimes \mcP^\theta_{1,x} \otimes \mcP^\theta_{2,x} \otimes \ldots \otimes \mcP^\theta_{m,x} \right)\rho_{RP_1...P_m}\right]
\end{equation*}
We then consider the $m$-player value $\omega_m(\mcR)$ of a referee configuration $\mcR$ to be the associated value of the $m$-player MoE game with referee configuration $\mcR$.
\begin{equation*}
    \omega_m(\mcR) = \omega(\mcG_m) = \sup\{p_{\mcG_m}(\mcS) : \mcS \text{ is a strategy for } \mcG_m\}
\end{equation*}
We consider one more distinguished value for a MoE game setup with referee configuration $\mcR$, which we call the \textit{classical} value $\omega_c(\mcR)$ of a referee configuration $\mcR$, which measures how well a player can guess the referee's measurement outcome without entanglement (i.e. using classical strategies).
\begin{equation*}
    \omega_c(\mcR) = \omega(\mcG_1,D=1) = \sup\{p_{\mcG_1}(\mcS) : \mcS \text{ is a $1$-bounded strategy for } \mcG_1\}
\end{equation*}
A referee configuration $\mcR$ is said to admit a $m$-player quantum advantage if $\omega_m(\mcR)>\omega_c(\mcR)$. Clearly, $\omega_m(\mcR)\geq \omega_c(\mcR)$, as any strategy without entanglement can be achieved with $m$ players (i.e. all $m$ players agree on their answers before the game begins), so a quantum advantage is only present when the players can use entanglement to their advantage, and outperform the classical value.

\begin{remark}
    In \cite[Example 4.2]{MOE:Johnston16}, a referee configuration consisting of four mutually unbiased basis measurements on a qutrit was found to demonstrate a slight 2-player quantum advantage over the classical value (there, the classical value is instead called the unentangled value). It is an open question as to whether a 2-player quantum advantage can be achieved with a question set of size three \cite{MOE:Russo17}. Further examples of 2-player quantum advantages in MoE games can be found in \cite{MOE:Dubois23, MOE:Moran25}. 
\end{remark}

Monogamy-of-Entanglement games provide a useful framework for understanding the limited ability of quantum correlations to be distributed in multi-partite systems. However, the framework is quite restrictive, in the sense that MoE games require \textit{simultaneous} consistent guesses among all the players. Rather, we may be interested in testing only the correlation between the referee's and \textit{individual} players' outcomes, motivating the study of OTG games. As suggested by its name, a One-at-a-Time Guessing Game (OTG game, for short), is a MRQG game in which only one player guesses the referee's measurement outcome. The player $P_i$ who makes the guesses is randomly selected by the referee with probability $q_i$. As an MRQG game, the evaluation procedure selects only one player at a time, so without loss of generality, we can write the selection function $q_I$ as a player selection function, i.e. $q_I(\{i\})=q_i$.

\begin{definition}\label{def:OTG_game}
    A One-at-a-Time Guessing Game is a Measurement-Result Quantum Guessing game  $\mcG$, equipped with a player selection probability mass function $(q_i)_{i=1}^m$ on $m$ players such that the player query set is $\mcI=\{\{1\},\{2\},\ldots,\{m\}\}$, and $q_I(\{i\})=q_i$,  i.e. the i-th player $P_i$ is selected with probability $q_i$ to guess the referee's measurement outcome. We write the evaluation procedure as $\mcE_q$ for OTG games.
\end{definition}

Following Definition~\ref{def:MoE_Strategy}, the winning probability of a strategy $\mcS$ for an OTG game $\mcG$ with player selection probability $q$ is given by
\begin{align*}
    p_{\mcG}(\mcS) &= \sum_{\theta \in \Theta, i \in [m]} p_\Theta(\theta)q_i \sum_{x \in \mcX} \Tr\left[\left(\mcM^\theta_x \otimes \bbI_{P_1} \otimes \ldots \otimes \mcP^\theta_{i,x} \otimes \ldots \otimes \bbI_{P_m}\right) \rho_{RP_1...P_m}\right]\\
    &= \sum_{\theta \in \Theta, i \in [m]} p_\Theta(\theta)q_i \sum_{x \in \mcX} \Tr\left[\left(\mcM^\theta_x \otimes \mcP^\theta_{i,x}\right) \rho_{RP_i}\right].
\end{align*}

We are again interested in characterizing the \textit{value} of an OTG game, which is the largest achievable winning probability of the game, $\omega(\mcG)$. In a similar fashion to MoE games, we would like to understand how, for a given referee configuration $\mcR$, the value of the OTG game changes with the player selection probability $q$. This motivates us to consider the associated value of an OTG game $\mcG = (\mcR,\mcE_q)$ with referee configuration $\mcR$ and player selection probability $q$:
\begin{equation*}
    \omega(\mcR,q) = \sup\{p_{\mcG}(\mcS) : \mcS \text{ is a strategy for } \mcG = (\mcR,\mcE_q)\}
\end{equation*}

There is a natural way to transform an $m$-player OTG game into an $m'$-player OTG game for $m<m'$: given an $m$-player OTG game with player selection probability $q$ on $[m]$, define a probability distribution $q'$ on $[m']$ by setting $q_i=q'_i$ for $i \in [m]$ and $q_i'=0$ otherwise. Operationally, the players $P_i$ for $i>m$ are never selected as guessers, so the $m'$-player game functions exactly like the corresponding $m$-player game. Consequently, we say that any $m$-player OTG game can be embedded into an $m'$-player OTG game for $m<m'$. This allows us to compare the values and strategies of OTG games with the same referee configuration but varying numbers of players.

We would also like to determine how well players can guess outcomes in OTG games with limited resources, just as we did for MoE games, so that we can identify when a non-trivial advantage is present. The first designated value we consider is the classical value $\omega_c(\mcR,q)$ of an OTG game $\mcG=(\mcR,\mcE_q)$, which is once again the players' largest achievable average winning probability without the use of entanglement.
\begin{equation*}
    \omega_c(\mcR,q) = \omega(\mcG,D=1) = \sup \{p_{\mcG}(\mcS) : \mcS \text{ is a 1-bounded strategy for } \mcG\}
\end{equation*}
It should not be too surprising that, for a given referee configuration $\mcR$, the classical values for the corresponding MoE game and any OTG game with player selection probability $q$ are identical, i.e. $\omega_c(\mcR)=\omega_c(\mcR,q)$. This follows easily, because the players can always agree upon guesses before the game begins. The second value of interest for OTG games is the maximum achievable average winning probability using states that decouple upon tracing out one player. Specifically, we designate the bipartite value $\omega_{bpe}(\mcR,q)$ to be the supremum winning probability over all strategies in which the state $\rho_{RP_1...P_m}$ satisfies $\Tr_{P_i}[\rho_{RP_1...P_m}] = \rho_R \otimes \Tr_{P_i}[\rho_{P_1...P_m}]$ for some player index $i \in [m]$.
\begin{align*}
    \omega_{bpe}(\mcR,q) = \sup\{p_{\mcG}(\mcS): \mcS &\text{ is a strategy for } \mcG \text{ such that }  \Tr_{P_i}[\rho_{RP_1...P_m}] = \rho_R \otimes \Tr_{P_i}[\rho_{P_1...P_m}]\\
    &\text{ for some player index } i \in [m]\}
\end{align*}

In particular, this includes strategies where one player is maximally entangled with the referee, leaving the remaining players unentangled with the referee. In fact, for any player $j\neq i$, the reduced density matrix for the referee and player $j$ is $\rho_R \otimes \rho_{P_{j}}$, which implies the measurement outcomes of the referee and player $j$ are independent. Consequently, the measurement statistics of player $j$ can be simulated with classical randomness, and the strategies $\mcS$ in the definition of $\omega_{bpe}(\mcR,q)$ may be restricted to those classical-quantum strategies with $d_j=1$ for $j \neq i$.  We will say that an OTG game $\mcG=(\mcR,\mcE_q)$ admits a quantum advantage if $\omega(\mcG)>\omega_c(\mcR,q)$, and furthermore an OTG game admits a super bipartite advantage if $\omega(\mcG)>\omega_{bpe}(\mcR,q)$.

Finally, we consider performance metric for a strategy $\mcS$, which is independent of the evaluation procedure $\mcE$ of the MRQG game $\mcG$, and by which we refer to as the \textit{minimum} winning probability $p_{\mcG,\min}(\mcS)$. It is defined to be the smallest winning probability for any individual player.
\begin{equation*}
    p_{\mcR,\min}(\mcS) = \min \left\{\sum_{\theta \in \Theta} p_\Theta(\theta) \sum_{x \in \mcX} \Tr\left[\left(\mcM^\theta_x \otimes \bbI_{P_1} \otimes \ldots \otimes \mcP^\theta_{i,x} \otimes \ldots \otimes \bbI_{P_m}\right) \rho_{RP_1...P_m}\right] : i \in [m]\right\}
\end{equation*}

The $m$-player \textit{maximum-minimum} value of a MRQG game (max-min value, for short) is defined as the largest achievable minimum winning probability over all $m$-player strategies $\mcS$.
\begin{equation*}
    \omega_{m,\max\min}(\mcR) = \sup \{p_{\mcR,\min}(\mcS) : \mcS \text{ is a strategy for } \mcG\}
\end{equation*}

\subsection{Relationship between Game Values}

With all definitions in order, we are ready to state the main results about general MoE games and OTG games. In particular, we demonstrate how optimal performance in these games changes with the number of players, the player selection probability, etc. We also compare the performance of optimal strategies to classical strategies and strategies using bipartite entanglement.

\begin{restatable}{theorem}{MoEValues}\label{thm:moe_values}
    Let $\mcG_m$ be a Monogamy-of-Entanglement game with referee configuration $\mcR = (\mcH_R,\Theta,p_\Theta,\mcX,\{\mcM^\theta\}_{\theta \in \Theta})$ and $m$ players. Then, the following holds:

    \begin{enumerate}
        \item If $\ell<m$, then $\omega_{m}(\mcR)\leq \omega_\ell(\mcR)$
        \item If $\underline{d},\underline{d}' \in \bbN^m$ are such that $\underline{d} \preceq \underline{d}'$, then $\omega_m(\mcR,\underline{d}) \leq \omega_m(\mcR,\underline{d}') \leq \omega_m(\mcR)$
        \item $\omega_m(\mcR)$ and $\omega_m(\mcR,D)$ can be approximated with randomized, symmetric strategies.
        \item $\frac{1}{|\mcX|} \leq \omega_m(\mcR,\underline{d})$ for all $\underline{d} \in \bbN^m$, and if $m\geq |\Theta|$, then $\omega_m(\mcR) =\omega_c(\mcR)$
        \item Classical-quantum strategies offer no advantage over classical strategies, i.e. $\omega_m(\mcR,\underline{d})=\omega_c(\mcR)$ if $d_i=1$ for some $i \in [m]$.
    \end{enumerate}
\end{restatable}

The proof of Theorem~\ref{thm:moe_values} is given in Appendix~\ref{sec:Appendix_Proofs}. The takeaway of Theorem~\ref{thm:moe_values} is that increasing the number of players limits the potential of a quantum advantage for the MoE game up until a point, beyond which there is no quantum advantage. The result that $\omega_m(\mcR)=\omega_c(\mcR)$ for $m\geq |\Theta|$ is a direct generalization of \cite[Theorem 4.1]{MOE:Johnston16}, where it is proven for $|\Theta|=2$.

The next theorem states some general results for OTG games, where special attention has been given to understanding the effect of the player selection probability on the value of the OTG game.

\begin{restatable}{theorem}{gameValues}\label{thm:game_values}

    Let $\mcG$ be a One-at-a-Time Guessing Games with referee configuration $\mcR = (\mcH_R,\Theta,p_\Theta,\mcX,\{\mcM^\theta\}_{\theta \in \Theta})$, and let $p,q$ be player selection probabilities on $m$ players. Then, the following holds:

    \begin{enumerate}
        \item If $p \preceq q$, then $\omega(\mcR,p) \leq \omega(\mcR,q)$ and $\omega(\mcG,p,D) \leq \omega(\mcG,q,D)$
        \item If $\underline{d},\underline{d}' \in \bbN^m$ are such that $\underline{d} \preceq \underline{d}'$, then $\omega(\mcR,q,\underline{d}) \leq \omega(\mcR,q,\underline{d}') \leq \omega(\mcR,q)$
        \item If $q_{\text{Unif}}$ is the uniform distribution on $[m]$, $\omega(\mcR,q_{\text{Unif}})$ and $\omega(\mcR,q_{\text{Unif}},D)$ can be approximated with randomized, symmetric strategies.
        \item There are OTG games which possess both a quantum advantage and a super bipartite advantage.
    \end{enumerate}

\end{restatable}

    A quick note is warranted for the last statement. In some OTG games, the existence of a quantum advantage may be trivial. For example, if the referee performs basis measurements using the BB84 states on $n$ qubits, the classical value of the corresponding OTG game with two players selected equally decays exponentially in $n$ (see \cite[Theorem 3]{MoE:Tomamichel13}), whereas maximum entanglement between the referee and a single players always allows for a winning probability of at least $\frac{1}{2}$. Thus, the more non-trivial result for OTG games comes from the potential of a super bipartite advantage.

    Finally, for a given referee configuration $\mcR$, one can relate the values of the various types of MRQG games, with MoE games and OTG games representing extreme ends of the spectrum:

\begin{restatable}{theorem}{compareValues}\label{thm:compare_vals}

    Let $\mcR = (\mcH_R,\Theta,p_\Theta,\mcX,\{\mcM^\theta\}_{\theta \in \Theta})$ be a referee configuration, and $q$ a player selection probability on $[m]$. Then,
    \begin{equation*}
        \frac{1}{|\mcX|} \leq \omega_m(\mcR) \leq \omega_{m,\max\min}(\mcR) = \omega(\mcR,q_{\text{Unif}}) \leq \omega(\mcR,q) \leq 1
    \end{equation*}

\end{restatable}

    Here, once again $q_{Unif}$ is the uniform distribution on $m$ players, and $q$ is an arbitrary player selection probability distribution. We leave the proofs of Theorem~\ref{thm:game_values} and Theorem~\ref{thm:compare_vals} to Appendix~\ref{sec:Appendix_Proofs}.


\section{Case Study: Pauli OTG Game}\label{sec:PauliOTG}

In this section, we consider the case study of a three Pauli-basis One-at-a-Time Guessing game with two players, which will serve as the foundation of our study of OTG games.

In the three Pauli-basis OTG game with two players, we rename the referee $(R)$ to Alice $(A)$, and the players $(P_1,P_2)$ to Bob $(B)$ and Charlie $(C)$. Alice's system consists of a single qubit $(\mcH_A = \bbC^2)$, and her measurements are selected with equal probability from one of the three measurements in the Pauli bases:
\begin{equation*}
    \begin{array}{ccc}
    \mcA^X = \bigl\{\ket{+}\bra{+},\ket{-}\bra{-}\bigr\}, ~~~ & \mcA^Y = \bigl\{\ket{+i}\bra{+i},\ket{-i}\bra{-i}\bigr\}, ~~~ & \mcA^Z = \bigl\{\ket{0}\bra{0},\ket{1}\bra{1}\bigr\}\\ 
    \mcA^X_0 = \ket{+}\bra{+} & \mcA^Y_0 = \ket{+i}\bra{+i} & \mcA^Z_0 = \ket{0}\bra{0}\\
    \mcA^X_1 = \ket{-}\bra{-} & \mcA^Y_1 = \ket{-i}\bra{-i} & \mcA^Z_1 = \ket{1}\bra{1}\\
\end{array}
\end{equation*}

Recall as in Definition~\ref{def:MoE_Strategy}, a strategy for the three Pauli-basis OTG game consists of a choice of Hilbert spaces $\mcH_B,\mcH_C$, a shared state $\rho_{ABC} \in \mcD(\mcH_A \otimes \mcH_{BC})$ and for each $\theta \in \Theta =  \{X,Y,Z\}$, and POVMs $\{\mcB^\theta_x\}_{x \in \mcX}$ and $\{\mcC^\theta_x\}_{x \in \mcX}$ for Bob and Charlie on their respective Hilbert spaces, with binary outcomes $x \in \mcX = \{0,1\}$.

When the player selection probability is $q=(q_B,q_C)$, the winning probability of the three Pauli-basis OTG game is
\begin{align*}
    p_{\mcG}(\mcS) &= \frac{1}{3}\sum_{\theta \in \{X,Y,Z\}} \sum_{x \in \{0,1\}}\left(q_B \Tr\left[\left(\mcA^\theta_x \otimes \mcB^\theta_x \otimes \bbI_C\right) \rho_{ABC}\right] + q_C \Tr\left[\left(\mcA^\theta_x \otimes \bbI_B \otimes \mcC^\theta_x\right) \rho_{ABC}\right]\right)\\
    &= \frac{1}{3}\sum_{\theta \in \{X,Y,Z\}} \sum_{x \in \{0,1\}}\left(q_B \Tr\left[\left(\mcA^\theta_x \otimes \mcB^\theta_x\right) \rho_{AB}\right] + q_C \Tr\left[\left(\mcA^\theta_x \otimes \mcC^\theta_x\right) \rho_{AC}\right]\right)
\end{align*}

\begin{figure}
    \centering
    \includegraphics[width=0.8\linewidth]{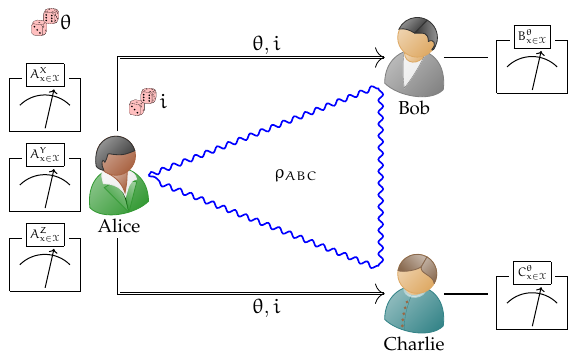}
    \caption{Pictorial representation of the two-player three Pauli-basis OTG game. A tripartite state $\rho_{ABC}$ is initally shared amongst all three participants. Alice randomly selects an Pauli observable $\theta \in \{X,Y,Z\}$ and a random player $i \in \{\text{Bob},\text{Charlie}\}$. She measures her qubit with observable $\theta$ and obtains a measurement outcome $x\in\{0,1\}$. The data $(\theta,i)$ is announced to both Bob and Charlie, and the selected player $i$ performs a measurement to obtain a measurement outcome $x'\in\{0,1\}$. The game is won if $x=x'$.}
    \label{fig:Pauli_OTG}
\end{figure}

\subsection{Equal Player Selection Probability}

We will first consider the three Pauli-basis OTG game with equal player-selection probabilities for Bob and Charlie (that is, $q=q_{\text{Unif}} = (\frac{1}{2},\frac{1}{2})$). Our first result characterizes the value of the three Pauli-basis OTG game.
\begin{theorem}\label{thm:equalPlayerSelection}
    Let $\mcG$ be the three Pauli-basis two player One-at-a-Time Guessing game with referee configuration $\mcR$ and equal player selection probability. Then, $\omega(\mcR,q_{\text{Unif}}) = \omega(\mcR,q_{\text{Unif}},2)=\frac{5}{6}$. Furthermore, $\mcG$ has both a quantum advantage and a super bipartite advantage.
\end{theorem}

One particular strategy used by Bob and Charlie to achieve a winning probability of $\frac{5}{6}$ is quite simple from the perspective of their measurements: they copy Alice's measurement up to outcome reordering:
\begin{equation*}
\begin{array}{ccc}
    \mcB^X_0 = \mcC^X_0 = \ket{+}\bra{+}, ~~~ & \mcB^Y_0 = \mcC^Y_0 = \ket{+i}\bra{+i}, ~~~ & \mcB^Z_0 = \mcC^Z_0 = \ket{1}\bra{1} \\
    \mcB^X_1 = \mcC^X_1 = \ket{-}\bra{-}, ~~~ & \mcB^Y_1 = \mcC^Y_1 = \ket{-i}\bra{-i}, ~~~ & \mcB^Z_1 = \mcC^Z_1 = \ket{0}\bra{0} \\
\end{array}
\end{equation*}

The state $\rho_{ABC}$ can then be chosen to be any eigenvector of $T$ corresponding to its maximum eigenvealue of $\frac{5}{6}$, where
\begin{equation*}
    T = \frac{1}{6} \sum_{\theta \in \{X,Y,Z\}} \sum_{x \in \{0,1\}} \mcA^\theta_x \otimes \mcB^\theta_x \otimes \bbI_C + \mcA^\theta_x \otimes \bbI_B \otimes \mcC^\theta_x 
\end{equation*}

Equivalently stated, we can choose $\rho_{ABC}$ to be any ground state of the operator $\bbI-T$. This space of vectors of eigenvalue $\frac{5}{6}$ is a two-dimsneional space spanned by two vectors:
\begin{equation*}
    \begin{array}{cc}
        \ket{W'_1} = \frac{1}{\sqrt{6}}\left(\ket{001}+\ket{010}+2\ket{100}\right), ~~~ & \ket{W'_2} = \frac{1}{\sqrt{6}}\left(2\ket{011}+\ket{101}+\ket{110}\right)   
    \end{array}
\end{equation*}

We note the resemblance of the state $\ket{W'_1}$ to the famous $W$-state on three qubits, which has the property that tracing out any single qubit leaves the remaining two qubits partially entangled: 
\begin{equation*}
    \ket{W} = \frac{1}{\sqrt{3}}\left(\ket{001}+\ket{010}+\ket{100}\right)
\end{equation*}

For the state $\ket{W'_1}$, tracing out Charlie's qubit also gives a partially entangled state between Alice and Bob:
\begin{equation*}
    \begin{array}{cc}
        \rho_{AB} = \Tr_C\bigl[ \ket{W'_1}\bra{W'_1}\bigr] = \frac{5}{6}\underbrace{\ket{\psi'}\bra{\psi'}}_{\text{Entangled}} + \frac{1}{6}\underbrace{\ket{00}\bra{00}}_{\text{Unentangled}}, ~~~ & \ket{\psi'} =  \frac{1}{\sqrt{5}}\left(\ket{01}+2\ket{10}\right) \\
    \end{array}
\end{equation*}

Furthermore, the mixed state $\sigma_{ABC} = \frac{1}{2}\ket{W_1'}\bra{W_1'} + \frac{1}{2}\ket{W_2'}\bra{W_2'}$ has the additional property of being the symmetric extension of a Werner state, up to a local unitary. Indeed,
\begin{equation*}
    \sigma_{AB} = \Tr_C\left[\sigma_{ABC}\right] = \frac{2}{3}\ket{\psi^+}\bra{\psi^+} + \frac{1}{3}\frac{\bbI_{AB}}{4}
\end{equation*}

where $\ket{\psi^+} = \frac{1}{\sqrt{2}}(\ket{01}+\ket{10})$ is maximally entangled. In fact, the Werner state $\sigma_{AB}$ saturates the below inequality for determining symmetric extendibility \cite{MOE:Chen14}.
\begin{equation*}
    \Tr[\sigma_B^2] \geq \Tr[\sigma_{AB}^2] - 4\sqrt{\det(\sigma_{AB})}
\end{equation*}
Consequently, it is impossible to increase the entanglement between Alice and Bob in the mixed state $\sigma_{ABC}$ (by increasing the coefficient of $\ket{\psi^+}\bra{\psi^+}$) without losing the symmetry $\sigma_{AB}=\sigma_{AC}$.

The optimality of this strategy has been verified by the second-level SDP of the NPA hierarchy, which can be found at \cite{Git:Schleppy26}. An overview of the NPA hierarchy can be found in Appendix~\ref{sec:Appendix_tech}. 

To demonstrate both the quantum advantage and the super bipartite advantage for the three Pauli-basis two-player OTG game with uniform player selection probability, we also computed the classical value and the bipartite value of $\mcG$.
\begin{itemize}
    \item $\omega_c(\mcR,q_{\text{Unif}}) = \frac{1}{2}+\frac{\sqrt{3}}{6} \approx 0.7887$
    \item $\omega_{bpe}(\mcR,q_{\text{Unif}}) = \omega(\mcR,q_{\text{Unif}},\underline{d}=(2,1)) = \frac{1}{2}+\frac{\sqrt{6}}{8} \approx 0.8062 > \omega_c(\mcR,q_{\text{Unif}})$
\end{itemize}

The above two values were computed by direct search. The classical value was obtained by searching over all $8$ deterministic strategies, in which the players respond deterministically with a bit $f(\theta)$ for each $\theta \in \{X,Y,Z\}$, and computing the largest eigenvalue of $T=\frac{1}{3}(\mcA^X_{f(X)} + \mcA^Y_{f(Y)} + \mcA^Z_{f(Z)})$. For this game, all classical strategies result in a winning probability of $\frac{1}{2}+\frac{\sqrt{3}}{6}$. The bipartite value was obtained by searching over all $8$ deterministic strategies $f(\theta)$ for Charlie, and for each of Charlie's strategies, computing Bob's optimal guessing probability via an SDP. Indeed, the guessing probability in this circumstance is given by
\begin{align*}
    p_\mcG(\mcS) &= \frac{1}{3}\sum_{\theta \in \{X,Y,Z\}} \sum_{x \in \{0,1\}} \frac{1}{2}\Tr\left[\left(\mcA^\theta_x \otimes \mcB^\theta_x\right) \rho_{AB}\right] + \frac{1}{2}\Tr\left[\left(\mcA^\theta_x \otimes \mcC^\theta_x\right) \rho_{AC}\right]\\
    &= \frac{1}{6}\sum_{\theta \in \{X,Y,Z\}} \sum_{x \in \{0,1\}}\Tr\left[\mcA^\theta_x \rho^\theta_x\right] + \frac{1}{6}\Tr\left[\left(\sum_{\theta \in \{X,Y,Z\}}\mcA^\theta_{f(\theta)}\right) \rho_{A}\right]
\end{align*}
where $\rho^\theta_x = \Tr_A\left[\left(\bbI_A \otimes \mcB^\theta_x\right)\rho_{AB}\right]$. Then, the SDP takes the form
\begin{align*}
    \text{maximize} ~~~ & \frac{1}{6}\sum_{\theta \in \{X,Y,Z\}} \sum_{x \in \{0,1\}}\Tr\left[\mcA^\theta_x \rho^\theta_x\right] + \frac{1}{6}\Tr\left[\left(\sum_{\theta \in \{X,Y,Z\}}\mcA^\theta_{f(\theta)}\right) \rho_{A}\right]\\
    \text{subject to} ~~~ &\rho_A,\rho^\theta_x \in \mcP(\mcH_A) ~~ \text{for all} ~\theta \in \{X,Y,Z\},~ x \in \{0,1\}\\
    & \sum_{x \in \{0,1\}} \rho^\theta_x = \rho_A ~~ \text{for all} ~\theta \in \{X,Y,Z\}\\
    & \Tr[\rho_A] = 1
\end{align*}

The solution to the SDP is then used to generate POVMs for Bob on a qubit. Once again, all classical deterministic strategies for Charlie give rise to the optimal value bipartite value $\frac{1}{2}+\frac{\sqrt{6}}{8}$. Taking the deterministic strategy $f(X)=f(Y)=f(Z)=0$ for Charlie, the optimal shared state between Alice and Bob determined from the SDP is (up to a local unitary on Bob's qubit)
\begin{equation*}
    \ket{\psi} = \cos\left(\frac{\pi}{8}\right)\ket{\phi_1} \otimes \ket{0} + \sin\left(\frac{\pi}{8}\right)\ket{\phi_2} \otimes \ket{1}
\end{equation*}
where $\ket{\phi_1}$ and $\ket{\phi_2}$ are the positive and negative eigenstates of $\frac{1}{\sqrt{3}}(X+Y+Z)=\ket{\phi_1}\bra{\phi_1} - \ket{\phi_2}\bra{\phi_2}$, respectively. It is curious to note that the optimal strategy with bipartite entanglement does not maximally entangle Alice and Bob. However, if that were the case, then $\rho_A=\frac{1}{2}\bbI_A$, and Charlie would always guess correctly with probability $\frac{1}{2}$. The optimal strategy with bipartite entanglement consequently trades-off entanglement between Alice and Bob, to reduce the classical uncertainty that Charlie sees.

Since the value of the three Pauli-basis OTG game with two players exceeds both the classical and bipartite values, we certify that it has both a quantum and super bipartite advantage.

\subsection{Varying Player Selection Probability}

In a fully general analysis of the two-player, three Pauli-basis OTG game, we aim to identify the optimal guessing probability $\omega(\mcR,q)$ for an arbitrary player selection probability distribution $q=(q_B,q_C)$. There are two reasons why we may consider an OTG game with asymmetric player selection probabilities. Firstly, it tackles directly the question of how much entanglement is necessary between the referee and individual players to play optimally, when the players have varying degrees of priority in the game; for instance, if Bob is selected with probability $q_B=0.75$ and Charlie with probability $q_C=0.25$, it is conceivable to assume that an optimal strategy may increase the entanglement between Alice and Bob, at the expense of the entanglement between Alice and Charlie (relative to the equal player selection probability case). Consequently, from the perspective of resource sharing, it is critical to understand how the player selection probability affects the game's optimal strategy, and how entanglement between the players should be distributed for optimal performance. Secondly, the premise of varying the player selection probability permits more broadly the study of multipartite correlations in quantum networks, and the underlying states that can achieve such correlations.

\begin{figure}
    \centering
    \includegraphics[width=0.8\linewidth]{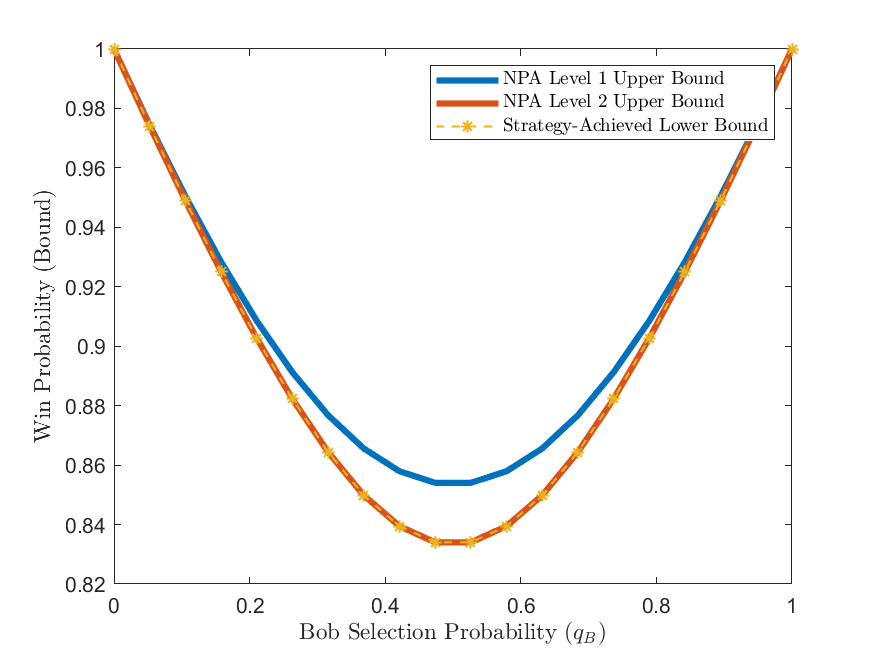}
    \caption{Bounds on the value $\omega_2(\mcG,q)$ of the two-player three Pauli-basis OTG game as a function of the player selection probability $q=(q_B,q_C)$. The lower bound from the parametrized qubit strategy agrees well with the upper bound from the second level SDP of the NPA hierarchy, up to numerical precision. The upper bound from the first level SDP of the NPA hierarchy is notably looser when the player selection probabilities are roughly equal.}
    \label{fig:PauliOTGPlot}
\end{figure}

Recall the expression for the winning probability of the three Pauli-basis OTG game when the player selection probability $q=(q_B,q_C)$ is arbitrary

\begin{equation*}
    p_{\mcG}(\mcS) = \frac{1}{3}\sum_{\theta \in \{X,Y,Z\}} \sum_{x \in \{0,1\}}\left(q_B \Tr\left[\left(\mcA^\theta_x \otimes \mcB^\theta_x\right) \rho_{AB}\right] + q_C \Tr\left[\left(\mcA^\theta_x \otimes \mcC^\theta_x\right) \rho_{AC}\right]\right)
\end{equation*}

We have the following characterization of the two-player three-basis OTG game with varying player selection probability.

\begin{figure}
    \centering
    \includegraphics[width=0.8\linewidth]{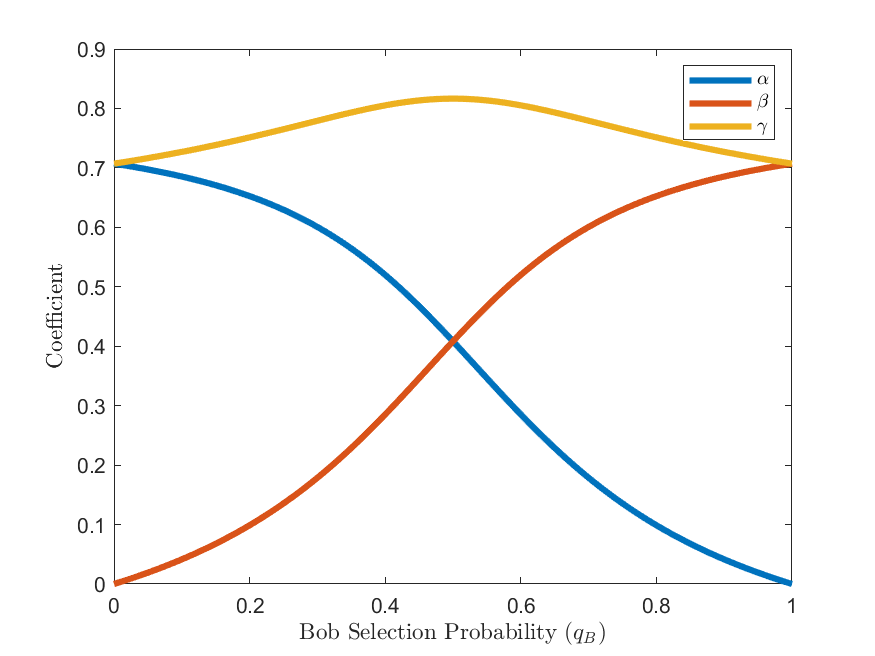}
    \caption{Coefficients of the parameterized optimal state $\ket{\psi}=\alpha\ket{001}+\beta\ket{010}+\gamma\ket{100}$ for the two-player three Pauli-basis OTG game, as a function of the player selection probability $q=(q_B,q_C)$.}
    \label{fig:PauliOTGCoeff}
\end{figure}

\begin{theorem}\label{thm:varPlayerSelection}
    Let $\mcG$ be the two-player three Pauli-basis game with player selection probability $q=(q_B,q_C)$. The value $\omega_2(\mcG,q)$ of the game is given by

    \begin{equation*}
        \omega_2(\mcG,q) = \frac{2}{3}+\frac{1}{6}\sqrt{1+3(q_B-q_C)^2}
    \end{equation*}

    Furthermore, the optimal winning probability is achieved for when the players have qubits, i.e. $\omega(\mcR,q,2)=\omega(\mcR,q)$.
\end{theorem}

Observe in the case of equal player selection probability $q_B=q_C=\frac{1}{2}$, the value of the game from Theorem~\ref{thm:varPlayerSelection} reduces to $\frac{5}{6}$, which is equal to the value obtained in Theorem~\ref{thm:equalPlayerSelection}. The optimal winning probability is obtained using a $W$-like state $\ket{W(q)}$
\begin{align*}
    &\ket{W(q)} = \alpha\ket{001} + \beta\ket{010} + \gamma\ket{100}&\\
    \alpha=\frac{2q_B - \sqrt{1+3(q_B-q_C)^2}}{2(q_B-q_C)}\gamma, ~~~~~& \beta=\frac{2q_C - \sqrt{1+3(q_B-q_C)^2}}{2(q_C-q_B)}\gamma, & \gamma=\sqrt{\frac{1}{3}+\frac{\sqrt{1+3(q_B-q_C)^2}}{3(1+3(q_B-q_C)^2)}.}
\end{align*}

The coefficients are plotted in Figure~\ref{fig:PauliOTGCoeff}. The state $\ket{W(q)}$ interpolates between states that are maximally entangled between Alice/Bob, and Alice/Charlie, namely $\ket{W(q)}=\frac{1}{\sqrt{2}}(\ket{010}+\ket{100})$ when $q_B=1$ and $\ket{W(q)}=\frac{1}{\sqrt{2}}(\ket{001}+\ket{100})$ when $q_C=1$. In the limit as $q_B,q_C\to \frac{1}{2}$, the state $\ket{W(q)}$ converges to the state $\ket{W'_1}$ in the previous section.

The optimality of $\omega(\mcR,q)$ has been certified by the second-level of the NPA hierarchy, up to numerical precision, which can be seen in Figure~\ref{fig:PauliOTGPlot}. The code used to compute the optimal values and plot Figures~\ref{fig:PauliOTGPlot} and~\ref{fig:PauliOTGCoeff} can be found at \cite{Git:Schleppy26}.



\section{Concluding Remarks and Broader Context}
In this work, we formulated One-at-a-Time Guessing games, a new class of extended non-local games which tests the strength of correlations between a characterized quantum referee and arbitrary non-communicating quantum players. As a Measurement-Result Quantum Guessing game, it lies at the extreme opposite end of games when compared to Monogamy-of-Entanglement games. The framework of One-at-a-Time Guessing games is flexible, as the player selection probability distribution can be modified to incorporate any number of players. The emphasis on bipartite correlations makes One-at-a-Time Guessing games more sensitive to entanglement-assisted advantages than Monogamy-of-Entanglement games, offering a new method to measure the strength of pairwise quantum correlations in quantum networks with a central node.

We proved a number of results regarding general Monogamy-of-Entanglement games and One-at-a-Time Guessing games, with special attention given to how the values of the games scales with increasing player pool and changing player selection probability. Crucially, for One-at-a-Time Guessing games, the value of the game for a given referee configuration respects the majorization order of the player selection probability distribution. This result is supplemented with a thorough analysis of a three Pauli-basis One-at-a-Time Guessing game with two players, in which we characterized the value of the game exactly as a function of the player selection probability distribution. A consequence of our analysis is that we extracted a family of states $\ket{W(q)}$ which interpolate between maximally entangled states (between Alice \& Bob vs. between Alice \& Charlie), and permit the players to guess optimally. When the players are selected equally, the optimal shared state can be taken to be a symmetric extension of a Werner state that saturates the symmetric extension inequality criterion. The family of states $\ket{W(q)}$ provide an insight into how tripartite entanglement may be optimally distributed to maximize the strengths of pairwise correlations with relative priority.

\begin{figure}[t]
    \centering
    \includegraphics[width=1\linewidth]{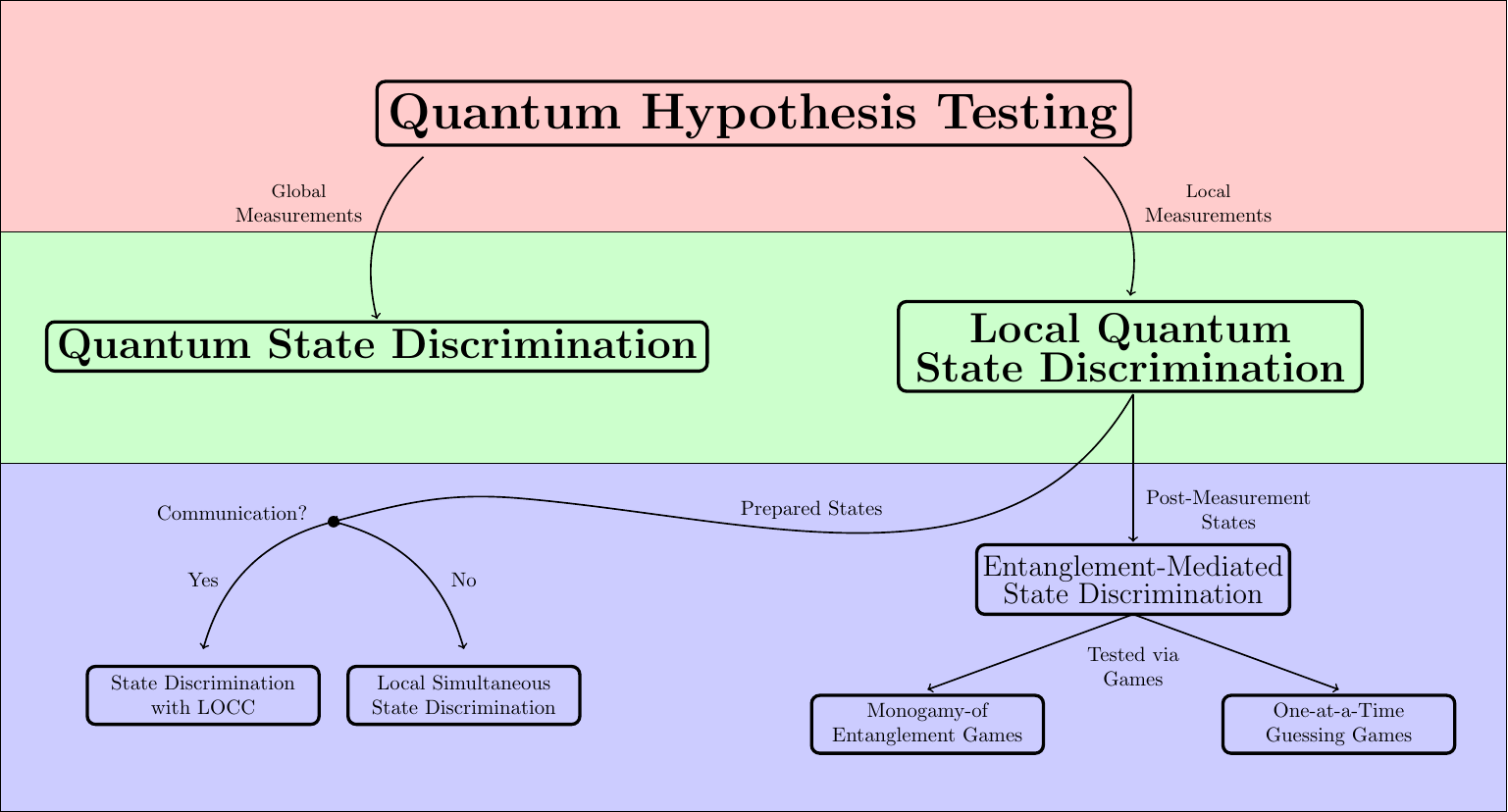}
    \caption{An illustration of various scenarios related to non-local state discrimination and hypothesis testing. Both Monogamy-of-Entanglement games and One-at-a-Time Guessing games can be viewed as examples of entanglement-mediated local hypothesis tests.}
    \label{fig:RelatedScenarios}
\end{figure}

A bit of commentary is in order for how One-at-a-Time Guessing games fit into the larger picture of local guessing and hypothesis testing. The idea of local, multi-partite state discrimination is not new and dates back to \cite{HT:Peres91}. In this setting, an ensemble of bipartite states $\{\rho_x\}_{x \in \mcX}$ are randomly distributed between two parties, who perform local measurements to obtain a guess for the label $x \in \mcX$. Several works along this line explored state discrimination with spatially-separated but communicating parties, for instance, in relation to quantum key distribution \cite{HT:Goldenberg95,Guess:Bennett1999,HT:Walgate00}. This framework would eventually be known as state discrimination with LOCC (local operations and classical communication). 

A useful weakening of state discrimination with LOCC is to exclude communication between the parties. This setting is known as Local Simultaneous State Discrimination, introduced in \cite{Guess:Majenz24}. Here, the parties are no longer strictly collaborating in the usual sense, since they are unable to agree upon a unified guess for the label $x\in \mcX$ via classical communication. Rather, the objective for the parties is to maximize the probability that they \textit{simultaneously} guess the label $x$. It is incorrect to say, however, that the situation reduces to independent instances of quantum state discrimination on each subsystem, even when the states are purely classical. A series of recent works \cite{Guess:Majenz24,Guess:Escola25,Guess:Pan22} have explored this topic more in depth, including separation of values (i.e. classical, quantum, non-signaling) and computational complexity issues.

What the current paper tackles, on the other hand, is slightly different. In both Monogamy-of-Entanglement games and One-at-a-Time Guessing games, the players seek to guess the outcome of a random measurement performed by a referee on their share of a multi-partite state. In Monogamy-of-Entanglement games, when a particular measurement is selected, the ensemble of post-measurement states on the players' subsystems gives rise to an instance of local simultaneous state discrimination, as noted in \cite{Guess:Majenz24}. What distinguishes to two scenarios, however, is that the selection of multi-partite state as part of the players' strategy also affects the post-measurement on the players' subsystems. Since the players' state discrimination task is fundamentally linked to how the referee's measurements steer the players' state, we will refer to this scenario broadly as entanglement-mediated local hypothesis tests.

The intricacy of both Monogamy-of-Entanglement games and One-at-a-Time guessing games arises from, among other things, the incompatibility of the referee's measurements, the limits of multi-partite entanglement due to Monogamy-of-Entanglement, and the induced local simultaneous state discrimination problem. In particular, the challenge of designing multi-partite states which allow players to share non-trivial, super-classical correlations with a referee makes  difficult. Restrictive tests such as Monogamy-of-Entanglement games are useful in showing the impossibility of various multi-partite correlations arising from any multi-partite state. On the contrary, One-at-a-Time Guessing games may serve as a valuable tool for identifying and characterizing multi-partite states which do indeed exhibit non-trivial, super-classical correlations when measured. This claim is substantiated by the family of optimal states $\ket{W(q)}$ that naturally arose from the optimal strategy of the three Pauli-basis One-at-a-Time Guessing game. Furthermore, for a particular set of referee measurements, it is possible that the optimal states for the One-at-a-Time Guessing game may aid the design of optimal states for the corresponding Monogamy-of-Entanglement games, and could be useful to identify Monogamy-of-Entanglement games with a non-trivial quantum advantage. Whether this is true, however, is unknown, and we leave this as an open question.



\bibliography{ref}
\bibliographystyle{IEEEtran}


\appendices

\section{Overview of Techniques}\label{sec:Appendix_tech}

This section will provide a brief summary of tools and techniques used throughout the proofs of the main results of this paper. The informed reader may choose to skip this section.

\subsection{Quantum State Discrimination}

Quantum state discrimination (also sometimes called quantum hypothesis testing) is a fundamental task in the theory of quantum information. Given a known collection of states $\rho_x \in \mcD(\mcH)$ occuring with probability $p_x$, the goal of quantum state discrimination is to determine a measurement on $\mcH$ that differentiates between the states as optimally as possible. There are two ways in which the discrimination may be optimal, either by minimizing the probability of misclassification error (i.e. a measurement declares $x$ when the true state is not $\rho_x$), or maximizing the probability of correct state identification subject to no misclassification error (by allowing for a `null' measurement outcome reflecting uncertainty of the true state). We wll focus on the former criterion.

Recall that any measurement statistics of a state $\rho \in \mcD(\mcH)$ can be modeled with a POVM $\mcM = \{\mcM_x\}_{x \in \mcX}$ on $\mcH$. The probability that a POVM $\mcM$ correctly identifies the label $x \in \mcX$ of an ensemble $(p_x,\rho_x)_{x \in \mcX}$ may be represented as
\[p_{\text{guess}} = \sum_{x \in \mcX} p_x \Tr\left[\mcM_x \rho_x\right]\]

where we will call $p_{\text{guess}}$ the guessing probability. To maximize the guessing probability for a given ensemble $(p_x,\rho_x)_{x \in \mcX}$, the following optimization problem must be considered.
\begin{align*}
    \text{maximize} & ~~~ \sum_{x \in \mcX} p_x \Tr\left[\mcM_x \rho_x\right]\\
    \text{subject to} & ~~~ \sum_{x \in \mcX} \mcM_x = \bbI\\
    & ~~~ \mcM_x \in \mcP(\mcH)
\end{align*}

The above optimization problem is simply a semidefinite program (albeit in non-standard form), and therefore can be solved efficiently. In case $|\mcX|=2$, there is a closed-form expression for the optimal guessing probability, given by the Holevo-Helstrom measurement \cite{HT:Helstrom69,HT:Holevo73}:
\[p_{\text{guess}}\leq \frac{1}{2} + \frac{1}{2} \left|\left| p_0\rho_0 - p_1\rho_1\right|\right|_1\]

where $||X||_1=\Tr\left[\sqrt{X^\dagger X}\right]$ is the Schatten 1-norm, equivalent to the sum of singular values of $X$, and $\mcX=\{0,1\}$ without loss of generality. Equality is achieved when the POVM elements $\mcM_0,\mcM_1$ are the projections onto the spaces spanned by the eigenvectors with positive and negative eigenvalues of $X=p_0\rho_0 - p_1\rho_1$, respectively (the projection onto the kernel of $X$ can be assigned arbitarily to the POVM elements).

\subsection{NPA Hierarchy Techniques}

We give a brief summary of techniques from the NPA hierarchy here, tailored towards our use of it for the three Pauli-basis OTG game. For a full treatment of the NPA hierarchy, refer to \cite{NPA:Navascues08,NPA:Pironio10}. The specific approach we will follow is similar to \cite{MOE:Botteron24}. Our numerical results from the use of SDPs from the NPA hierarchy can be found at \cite{Git:Schleppy26}. 

Recall the expression for the winning probability of the three Pauli-basis OTG game.
\begin{align*}
    p_{\mcG}(\mcS) &= \frac{1}{3}\sum_{\theta \in \{X,Y,Z\}} \sum_{x \in \{0,1\}}\left(q_B \Tr\left[\left(\mcA^\theta_x \otimes \mcB^\theta_x \otimes \bbI_C\right) \rho_{ABC}\right] + q_C \Tr\left[\left(\mcA^\theta_x \otimes \bbI_B \otimes \mcC^\theta_x\right) \rho_{ABC}\right]\right)\\
    &= \frac{1}{3}\sum_{\theta \in \{X,Y,Z\}} \sum_{x \in \{0,1\}}\left(q_B \Tr\left[\left(\mcA^\theta_x \otimes \mcB^\theta_x\right) \rho_{AB}\right] + q_C \Tr\left[\left(\mcA^\theta_x \otimes \mcC^\theta_x\right) \rho_{AC}\right]\right)
\end{align*}

The value of the OTG game $\omega(\mcR,q)$ is found by optimizing $p_{\mcG}(\mcS)$ over all strategies. Note that we may assume that all POVMs are PVMs, and the the shared state $\rho_{ABC}$ is a pure state $\ket{\psi}\bra{\psi}$, following a standard purification argument and application of Naimark's theorem as in \cite{MoE:Tomamichel13}. Unfortunately, $p_{\mcG}(\mcS)$ is difficult to optimize for two reasons: 1) it is non-convex in the strategy $\mcS$, due to the appearance of products of variables (e.g. $\mcB^\theta_x$ and $\rho_{AB}$), and 2) there is no upper bound on the dimension of Bob and Charlie's Hilbert spaces.

To manage this, we resort to constraint relaxations in order to recover convexity. Firstly, we will rewrite $p_{\mcG}(\mcS)$ in terms of the \textit{observables} measured by the referee and all players:
\begin{align*}
p_{\mcG}(\mcS) &= \frac{1}{3}\sum_{\theta \in \{X,Y,Z\}} \sum_{x \in \{0,1\}} q_B \Tr\left[ \left(\frac{\bbI_A + (-1)^x \mcA^\theta}{2} \otimes \frac{\bbI_B + (-1)^x \mcB^\theta}{2}\right) \rho_{AB}\right]\\
 &~~~~~~~~~~~~~~~~~~~~+ q_C \Tr\left[ \left(\frac{\bbI_A + (-1)^x \mcA^\theta}{2} \otimes \frac{\bbI_C + (-1)^x \mcC^\theta}{2}\right) \rho_{AC}\right]\\
 &= \frac{1}{2} + \frac{1}{6}\sum_{\theta \in \{X,Y,Z\}} q_B\Tr\left[ \left(\mcA^\theta \otimes \mcB^\theta\right) \rho_{AB}\right] + q_C\Tr\left[\left(\mcA^\theta \otimes \mcC^\theta\right)\rho_{AC}\right]\\
 &= \frac{1}{2} + \frac{1}{6}\sum_{\theta \in \{X,Y,Z\}} q_B \bra{\psi}\mcA^\theta \otimes \mcB^\theta \otimes \bbI_C \ket{\psi} + q_C \bra{\psi}\mcA^\theta \otimes \bbI_B \otimes \mcC^\theta \ket{\psi}
\end{align*}

where $\mcA^\theta = \mcA^\theta_0 - \mcA^\theta_1$, $\mcB^\theta = \mcB^\theta_0 - \mcB^\theta_1$, and $\mcC^\theta = \mcC^\theta_0 - \mcC^\theta_1$ are all unitary matrices squaring to identity. Crucially, observe that $p_{\mcG}(\mcS)$ is an affine function in terms of the terms $\bra{\psi}\mcA^\theta \otimes \mcB^\theta \otimes \bbI_C \ket{\psi}$ and $\bra{\psi} \mcA^\theta \otimes \bbI_B \otimes \mcC^\theta \ket{\psi}$ (i.e expectations of the joint observables). Therefore, if we could determine necessary semidefinite constraints on the joint observable expectations, we can upper bound $p_{\mcG}(\mcS)$ using a semidefinite program. This is exactly what the NPA hierarchy of semidefinite programs attempts to do.

From here on, we write $\mcA^\theta = \mcA^\theta \otimes \bbI_B \otimes \bbI_C$, $\mcB^\theta = \bbI_A \otimes \mcB^\theta \otimes \bbI_C$, and $\mcC^\theta = \bbI_A \otimes \bbI_B \otimes \mcC^\theta$, and note that observables on distinct subsystems necessarily commute with each other (e.g. $\mcA^\theta \mcB^{\theta'} = \mcB^{\theta'}\mcA^\theta)$. Recall also that we have fixed Alice's observables to be the Pauli observables (i.e. $\mcA^X=X,\mcA^Y=Y,\mcA^Z=Z)$. Now, we form an alphabet $\Gamma$ whose symbols are comprised of products of observables from a \textit{formal} strategy $\mcS$ (that is, we do not assume any relationships among Bob's observables or among Charlie's observables). In our setup, we will take the alphabet 
\[\Gamma = \{\mcA^X\mcB^X,\mcA^Y\mcB^Y,\mcA^Z\mcB^Z,\mcA^X\mcC^X,\mcA^Y\mcC^Y,\mcA^Z\mcC^Z\}.\]

Next, we define a word $W$ to be any finite product of symbols from our alphabet $\Gamma$, e.g. $W=(\mcA^X\mcB^X)(\mcA^Y\mcC^Y)(\mcA^X\mcC^X)$, and we denote by $\Gamma_{\leq k}$ the set of all words of length at most $k$. Note that the empty string of length zero will be included in $\Gamma_{\leq k}$.

For any formal strategy $\mcS$, we may define the following $k^{\text{th}}$-level correlation matrix $Q_k \in \bbC^{|\Gamma_{\leq k}| \times |\Gamma_{\leq k}|}$: For every word $W \in \Gamma_k$, define the vector $\ket{W}=W\ket{\psi}$. Then, $Q_k$ will simply be the Gram matrix associated with the vectors $\{\ket{W}:W \in \Gamma_k\}$, i.e.
\[(Q_k)_{W,V} = \bra{W}\ket{V} = \bra{\psi}W^\dagger V\ket{\psi}.\]

There are several properties that $Q_k$ will automatically satisfy due to its construction. Firstly, since $Q_k$ is a Gram matrix, it is necessarily positive semidefinite (i.e. $Q_k \in \mcP(\bbC^{|\Gamma_{\leq k}|}))$. Secondly, denoting the empty string by $\bbI$, the normalization of the state $\ket{\psi}$ implies $(Q_k)_{\bbI,\bbI}=\bra{\psi}\ket{\psi}=1$. Lastly, due to the commutation of observables on distinct subsystems, as well as the anticommutation of some observables on Alice's subsystem, it may occur that some words are actually equivalent to each other up to sign. For instance, one has $(\mcA^X\mcB^X)(\mcA^X\mcC^X) = (\mcA^X\mcC^X)(\mcA^X\mcB^X)$, but $(\mcA^Y\mcB^Y)(\mcA^Z\mcC^Z)=-(\mcA^Z\mcC^Z)(\mcA^Y\mcB^Y)$. Correspondingly, $Q_k$ must necessarily satisfy some linear constraints in its entries.

We make this idea precise here. We say that two words $W,V$ are equivalent (and write $W\sim V$) if for every assignment of $\mcB^\theta$, $\mcC^\theta$ to unitary matrices squaring to identity, the equality $W=V$ holds. For instance, one has $(\mcA^X\mcB^X)(\mcA^X\mcC^X) \sim (\mcA^X\mcC^X)(\mcA^X\mcB^X)$ since $\mcA^X$, $\mcB^X$ and $\mcC^X$ commute, but $(\mcA^Y\mcB^Y)(\mcA^Z\mcB^Z)\not \sim (\mcA^Z\mcB^Z)(\mcA^Y\mcB^Y)$, since equality is not guarenteed under any assignment of operators $\mcB^Y$ and $\mcB^Z$. More generally, for any finite sets of words $\{W_i\}_i$ and $\{V_j\}_j$, we say that the formal linear combinations $\sum_i \alpha_i W_i$ and $\sum_j \beta_jV_j$ are equivalent (and write $\sum_i\alpha_iW_i \sim \sum_j \beta_j V_j$) if for every assignment of operators $\mcB^\theta,\mcC^\theta$, the equality $\sum_i \alpha_i W_i=\sum_j\beta_jV_j$ holds. Consequently, equivalence of formal linear combinations of words naturally imposes constraints on the correlation matrix $Q_k$
\begin{align*}
    \sum_i \alpha_i W^\dagger_i V_i \sim \sum_j \beta_j \tilde{W}^\dagger_j \tilde{V}_j \implies &\sum_i \alpha_i(Q_k)_{W_i,V_i} =  \bra{\psi}\left(\sum_i \alpha_i W^\dagger_i V_i\right)\ket{\psi}\\ 
    & = \bra{\psi}\left(\sum_j \beta_j \tilde{W}^\dagger_j \tilde{V}_j\right)\ket{\psi} = \sum_j \beta_j(Q_k)_{\tilde{W}_j,\tilde{V}_j}
\end{align*}
Constraints of this form are called \textit{NPA constraints}, and we will write $Q_k \in \text{NPA}_k$ to indicate that the equality $\sum_i \alpha_i (Q_k)_{W_i,V_i} = \sum_j \beta_j (Q_k)_{\tilde{W}_j,\tilde{V}_j}$ holds whenever $\sum_i \alpha_i W^\dagger_i V_i \sim \sum_j \beta_j \tilde{W}^\dagger_j \tilde{V}_j$ for any finite sets of words $\{W_i\}_i$, $\{V_i\}_i$, $\{\tilde{W}_j\}_j$, $\{\tilde{V}_j\}_j$ with word lengths at most $k$, in addition to state normalization. It follows immediately that $\text{NPA}_k$ is a non-empty affine subspace of $\bbC^{|\Gamma_{\leq k}|\times|\Gamma_{\leq k}|}$. Note furthermore that the NPA constraints necessarily require the diagonal entries of $Q_k$ to be one, since all words are necessarily unitary and therefore $Q_{W,W} = \bra{\psi}W^\dagger W\ket{\psi} = 1$.

Finally, to recover $p_{\mcG}(\mcS)$ from $Q_k$, we observe that $p_{\mcG}(\mcS)$ is nothing but an affine combination of the entries of $Q_k$, specifically $(Q_k)_{\bbI,W}$ and $(Q_k)_{V,\bbI}$ for $V,W \in \Gamma$. Consquently, we can define a Hermitian matrix $A \in \bbC^{|\Gamma_{\leq k}|\times|\Gamma_{\leq k}|}$ such that $p_{\mcG,q}(\mcS)=\frac{1}{2}+\frac{1}{6}\Tr[AQ_k]$, where
\[A = \begin{cases}
    A_{\bbI,W}=A_{V,\bbI} = \frac{1}{2}q_B ~~ \text{ for } V,W \in \{\mcA^X\mcB^X,\mcA^Y\mcB^Y,\mcA^Z\mcB^Z\}\\
    A_{\bbI,W}=A_{V,\bbI} = \frac{1}{2}q_C ~~ \text{ for } V,W \in \{\mcA^X\mcC^X,\mcA^Y\mcC^Y,\mcA^Z\mcC^Z\}\\
\end{cases}\]

Combining all constraints together, we arrive at the $k$-th level semidefinite program $S_k$ for the NPA hierarchy of the game $\mcG$:
\begin{align*}
    \text{maximize} ~~~  &p_{\mcG}(\mcS)=\frac{1}{2}+\frac{1}{6}\Tr[AQ_k]\\
    \text{subject to} ~~~ & Q_k \in \mcP(\bbC^{|\Gamma_{\leq k}|})\\
    & Q_k \in \text{NPA}_k
\end{align*}

The semidefinite program $S_k$ provides an upper bound for the value $\omega(\mcR,q)$, with the bound getting tighter as $k\to \infty$. The full NPA hierarchy (as in \cite{NPA:Navascues08,NPA:Pironio10}) converges to the commuting value of the OTG game $\mcG$, which in general, need not be equal to the value $\omega(\mcR,q)$.


\section{Proof of Results from Section~\ref{sec:OTG_Games}}\label{sec:Appendix_Proofs}
In this section, we prove the results of Section~\ref{sec:OTG_Games}. We will restate all of the theorems for convenience.

\MoEValues*

\begin{proof}
    We will prove the claims in order.
    
    \textbf{(1):} Suppose $\mcG_m$ is the MoE game with referee configuration $\mcR=(\mcH_R,\Theta,p_\Theta,\mcX,\{\mcM^\theta\}_{\theta \in \Theta})$ and $m$ players. The condition that $m$ players simultaneously guess the referee' measurement outcome is more restrictive than $\ell$ players simultaneously guessing the referee's measurement outcome, if $\ell<m$. We make this statement more precise. If $\ell<m$ and $\mcS_m$ is a strategy for the $m$-player MoE game $\mcG_m$, we can form a strategy $\mcS_{\ell}$ for the $\ell$-player MoE game $\mcG_{\ell}$ by simply disregarding the last $m-\ell$ players.
\begin{align*}
    \mcS_m &= (\mcH_{P_1...P_{m}},\rho_{RP_1...P_{m}},\{\mcP^\theta_i\}_{\theta \in \Theta} : i \in [m])\\
    \implies \mcS_{\ell}&=(\mcH_{P_1...P_{\ell}},\rho_{RP_1...P_{\ell}},\{\mcP^\theta_i\}_{\theta \in \Theta} : i \in [\ell])
\end{align*}

The probability of winning the $\ell$-player MoE game $\mcG_{\ell}$ under strategy $\mcS_{\ell}$ is
\begin{align*}
    p_\mcG(\mcS_\ell) &= \sum_{\theta \in \Theta} p_\Theta(\theta) \sum_{x \in \mcX} \Tr\left[\left(\mcM^\theta_x \otimes \mcP^\theta_{1,x} \otimes \ldots \otimes \mcP^\theta_{\ell,x} \right)\rho_{RP_1...P_{\ell}}\right]\\
    &= \sum_{\theta \in \Theta} p_\Theta(\theta) \sum_{x \in \mcX} \Tr\left[\left(\mcM^\theta_x \otimes \mcP^\theta_{1,x} \otimes \ldots \otimes \mcP^\theta_{\ell,x} \otimes \bbI_{P_{\ell+1}} \otimes \ldots \otimes  \bbI_{P_m}\right)\rho_{RP_1...P_{m}}\right]\\
    &\geq \sum_{\theta \in \Theta} p_\Theta(\theta) \sum_{x \in \mcX} \Tr\left[\left(\mcM^\theta_x \otimes \mcP^\theta_{1,x} \otimes \ldots \otimes \mcP^\theta_{\ell,x} \otimes \mcP^\theta_{\ell+1,x} \otimes \ldots \otimes  \mcP^\theta_{m,x}\right)\rho_{RP_1...P_{m}}\right]\\
    &= p_\mcG(\mcS_m).
\end{align*}

Taking the supremum over all strategies $\mcS_m$ for the $m$-player MoE game $\mcG_m$ proves that $\omega_m(\mcR)\leq \omega_\ell(\mcR)$ whenever $\ell<m$. 

\textbf{(2):} Any $\underline{d}$-bounded strategy for the MoE game $\mcG_m$ is automatically a $\underline{d}'$-bounded strategy for $\mcG_m$ whenever $\underline{d}\preceq \underline{d}'$, by definition. Consequently, $p_{\mcG_m}(\mcS) \leq \omega_m(\mcR,\underline{d}')$ for any $\underline{d}$-bounded strategy $\mcS$, and taking the supremum over all $\underline{d}$-bounded strategies $\mcS$ proves that $\omega_m(\mcR,\underline{d}) \leq \omega_m(\mcR,\underline{d}')$ whenever $\underline{d}\preceq \underline{d}'$. Trivially, any $\underline{d}'$-bounded strategy for $\mcG_m$ is also just an ordinary strategy for $\mcG_m$, so by the same reasoning $\omega_m(\mcR,\underline{d})\leq \omega_m(\mcR)$.

\textbf{(3):} Consider any $D$-bounded $\mcS$ for $\mcG_m$. We construct a randomized strategy $\mcS^{\text{rand}}$ over all permutations of the players' subsystems:
\begin{align*}
    \mcS&=(\mcH_{P_1...P_m},\rho_{RP_1...P_m},\{\mcP^\theta_i\}_{\theta \in \Theta}: i \in [m])\\
    \mcS^{\text{rand}} &= (C_{P_1...P_m} \otimes\mcH_{P_1...P_{m}},\tilde{\rho}_{RP_1...P_{m}},\{\tilde{\mcP}^\theta_i\}_{\theta \in \Theta} : i \in [m])
\end{align*}

where $C_{P_i}$ is a classical register with $m$ possible states $\{\ket{j}\bra{j}\}_{j=1}^m$, and both the shared state $\tilde{\rho}_{RP_1...P_m}$ and player POVMs $\{\tilde{\mcP}^\theta_i\}_{\theta \in \Theta}$ are constrcuted from $\mcS$.
\begin{align*}
\tilde{\rho}_{RP_1...P_m} &= \frac{1}{m!}\sum_{\sigma \in S_m} \ket{\sigma(1)}\bra{\sigma(1)}_{P_1} \otimes \ket{\sigma(2)}\bra{\sigma(2)}_{P_2} \otimes \ldots \otimes \ket{\sigma(m)}\bra{\sigma(m)}_{P_m} \otimes\rho_{RP_{\sigma(1)}...P_{\sigma(m)}}\\
\tilde{\mcP}^\theta_{i,x} &= \sum_{j=1}^m \ket{j}\bra{j}_{P_i} \otimes \mcP^\theta_{j,x}
\end{align*}
Notice that $\tilde{\rho}_{RP_1...P_m}$ is invariant under reordering the players, and that the POVM elements $\tilde{\mcP}^\theta_i$ are the same for each player under relabeling the classical register $P_i$. Also, note that since all players' subsystems have dimension at most $D$, permuting the players' POVMs does not make any players' subsystem exceed the dimension bound $D$. Thus, the strategy is symmetric and $D$-bounded, so it suffices to show that the strategy $\mcS^{\text{rand}}$ achieves the same winning probability as the strategy $\mcS$.
\begin{align*}
    p_{\mcG_m}(\mcS^{\text{rand}}) &= \sum_{\theta \in \Theta} p_{\Theta}(\theta) \sum_{x \in \mcX} \Tr\left[\left(\mcM^\theta_x \otimes \tilde{\mcP}^\theta_{1,x} \otimes \ldots \otimes \tilde{\mcP}^\theta_{m,x}\right)\tilde{\rho}_{RP_1...P_m}\right]\\
    &= \frac{1}{m!} \sum_{\sigma \in S_m} \sum_{j \in [m]^m} \sum_{\theta \in \Theta} p_{\Theta}(\theta) \sum_{x \in \mcX} \Tr\Big[\Big( \ket{j_1}\bra{j_1} \otimes \ldots \otimes \ket{j_m}\bra{j_m} \otimes \mcM^\theta_x \otimes \mcP^\theta_{j_1,x} \otimes \ldots \otimes \mcP^\theta_{j_m,x}\Big) \\
    & ~~~~~~~~~~~~~~~~~~~~~~~~~~~~~~~~~~~~~~~~~~~~~~\Big(\ket{\sigma(1)}\bra{\sigma(1)} \otimes \ldots \otimes \ket{\sigma(m)}\bra{\sigma(m)} \otimes \rho_{RP_{\sigma(1)}...P_{\sigma(m)}}\Big)\Big] \\
    &= \frac{1}{m!} \sum_{\sigma \in S_m} \sum_{\theta \in \Theta} p_{\Theta}(\theta) \sum_{x \in \mcX} \Tr\left[\left(\mcM^\theta_x \otimes \mcP^\theta_{\sigma(1),x} \otimes \ldots \otimes \mcP^\theta_{\sigma(m),x}\right)\rho_{RP_{\sigma(1)}...P_{\sigma(m)}}\right] \\
    &= \sum_{\theta \in \Theta} p_{\Theta}(\theta) \sum_{x \in \mcX} \Tr\left[\left(\mcM^\theta_x \otimes \mcP^\theta_{1,x} \otimes \ldots \otimes \mcP^\theta_{m,x}\right)\rho_{RP_{1}...P_{m}}\right] \\
    &= p_{\mcG_m}(\mcS)
\end{align*}

Taking the supremum over all strategies $\mcS$ shows that randomized, symmetric $D$-bounded strategies for $\mcG_m$ can approximate $\omega_m(\mcR,D)$. In fact, due to the compactness of POVMs on a $D$-dimensional space, the value $\omega_m(\mcR,D)$ can be obtained exactly. Furthermore, since every strategy $\mcS$ is a $D$-bounded strategy for sufficiently large $D$ ($D \geq \max_i \dim(\mcH_{P_i})$), then the same argument can be used to show that randomized, symmetric strategies for $\mcG_m$ can approximate $\omega_m(\mcR)$.

\textbf{(4):} Let us consider the randomized strategy $\mcS$ for $\mcG_m$ with shared state $\rho_{RP_1...P_m} = \frac{1}{|\mcX|}\sum_{x \in \mcX} \ket{x}\bra{x}^{\otimes m}_{P_1...P_m} \otimes \frac{1}{d_{RP_1...P_m}}\bbI_{RP_1...P_m} \in \mcD(C_{P_1...P_m} \otimes \mcH_{RP_1...P_m})$, with the classical register $C_{P_i}$ being on $|\mcX|$ states $\{\ket{x}\bra{x}\}_{x \in \mcX}$. The players' POVMs will be identical, $\mcP^\theta_{i,x} = \sum_{y \in \mcX} \ket{y}\bra{y} \otimes \delta_{xy}\bbI_{P_i}$. When Alice selects $\theta \in \Theta$, the probability of simultaneous guessing is then
\begin{align*}
    \Pr\left[\text{Win}|\theta\right] &= \sum_{x \in \mcX} \Tr\left[\left(\mcM^\theta_x \otimes \mcP^\theta_{1,x} \otimes \ldots \otimes \mcP^\theta_{m,x}\right)\rho_{RP_1...P_m}\right]\\
    &= \frac{1}{d_{RP_1...P_m}|\mcX|}\sum_{x,z \in \mcX} \sum_{y \in \mcX^m} \Tr\left[\left(\ket{y}\bra{y} \otimes \mcM^\theta_x \otimes \delta_{xy_1}\bbI_{P_1} \otimes \ldots \otimes \delta_{xy_m}\bbI_{P_m}\right)\left(\ket{z}\bra{z}^{\otimes m} \otimes \bbI_{RP_1...P_m}\right)\right]\\
    &= \frac{1}{d_{RP_1...P_m}|\mcX|} \sum_{x \in \mcX} \Tr\left[ \mcM^\theta_x \otimes \bbI_{P_1...P_m}\right]\\
    &= \frac{1}{d_{RP_1...P_m}|\mcX|} \Tr\left[ \bbI_R \otimes \bbI_{P_1...P_m}\right]\\
    &= \frac{1}{|\mcX|}
\end{align*}

Thus $p_{\mcG_m}(\mcS)=\frac{1}{|mcX|}$ and therefore $\frac{1}{|\mcX|} \leq \omega_m(\mcR,\underline{d})$. Now, assume that $m\geq |\Theta|$. We will show that $\omega_m(\mcR)=\omega_c(\mcR)$ by generalizing the technique used in \cite[Theorem 4.1]{MOE:Johnston16}. Since $\omega_{c}(\mcR) \leq \omega_{\ell}(\mcR) \leq \omega_{m}(\mcR)$ whenever $m\leq \ell$, it suffices to prove the claim for $m=|\Theta|$. Take any bijective function $f:\Theta \to [m]$. Then, for any strategy $\mcS$, the following holds,
\begin{align*}
    p_{\mcG_m}(\mcS) &= \sum_{\theta \in \Theta}p_\Theta(\theta) \sum_{x \in \mcX} \Tr\left[\left(\mcM^\theta_x \otimes \mcP^\theta_{1,x} \otimes \ldots \otimes \mcP^\theta_{f(\theta),x} \otimes \ldots \otimes \mcP^\theta_{m,x} \right)\rho_{RP_1...P_m}\right]\\
    &\leq \sum_{\theta \in \Theta}p_\Theta(\theta) \sum_{x \in \mcX} \Tr\left[\left(\mcM^\theta_x \otimes \bbI_{P_1} \otimes \ldots \otimes \mcP^\theta_{f(\theta),x} \otimes \ldots \otimes  \bbI_{P_m} \right)\rho_{RP_1...P_m}\right]\\
    &= \sum_{\theta \in \Theta}p_\Theta(\theta) \sum_{x \in \mcX} \Tr\left[\left(\mcM^\theta_x \otimes \mcP^\theta_{f(\theta),x} \right)\rho_{RP_{f(\theta)}}\right]
\end{align*}
Now, we construct a randomized classical strategy $\mcS_c$ which achieves this upper bound for $p_{\mcG_m}(\mcS)$. Let $\sigma_{RP_1...P_m} = \sum_{y \in \mcX^m} p_y\sigma_{R,y}  \otimes \ket{y}\bra{y}^{\otimes m}\in \mcD(\mcH_R \otimes C_{P_1...P_m})$ be a classical-quantum state, where $C_{P_i}$ is a classical register with $|\mcX|$ states, $p_y=\Tr\left[\left(\bbI_R \otimes \mcP^{f^{-1}(1)}_{1,y_1} \otimes \ldots \otimes \mcP^{f^{-1}(m)}_{m,y_m}\right)\rho_{RP_1...P_m}\right]$ is the joint probability that player $i$ receives the outcome $y_i \in \mcX$ upon measuring with POVM $\mcP^{f^{-1}(i)}_i$, and $\sigma_{R,y} = \frac{1}{p_y}\Tr_{P_1...P_m}\left[\left(\bbI_{R} \otimes \mcP^{f^{-1}(1)}_{1,y_1} \otimes \ldots \otimes \mcP^{f^{-1}(m)}_{m,y_m}\right)\rho_{RP_1...P_m}\right]$ is the post-measurement state on the referee's subsystem when the sequence of player measurement outcomes is $y=(y_1,\ldots,y_m)$. Observe that the state $\sigma_{RP_1...P_m}$ is seperable between the referee and the players, so therefore it is not entangled. Each player has access to one share of the parameter $y \in \mcX^m$ in the state $\sigma_{RP_1...P_m}$. The players measure their classical register with the POVM $\tilde{\mcP}^{\theta}_{i,x} = \sum_{y \in \mcX^m} \delta_{xy_{f(\theta)}}\ket{y}\bra{y}$. The new winning probability under this classical strategy $\mcS_c$ is then,
\begin{align*}
    p_{\mcG_m}(\mcS_c) &= \sum_{\theta \in \Theta}p_\Theta(\theta) \sum_{x \in \mcX} \Tr\left[\left(\mcM^\theta_x \otimes \tilde{\mcP}^\theta_{1,x} \otimes \ldots \otimes \tilde{\mcP}^\theta_{m,x} \right)\sigma_{RP_1...P_m}\right]\\
    &= \sum_{y,y^{(1)},\ldots,y^{(m)} \in \mcX^m} ~~\sum_{\theta \in \Theta}p_\Theta(\theta) \sum_{x \in \mcX} \delta_{xy^{(1)}_{f(\theta)}}\ldots\delta_{xy^{(m)}_{f(\theta)}}\\
    & ~~~~~~~~~~~~~~~~~~~~~~~~~~~~~~~~~~\Tr\left[\left(\mcM^\theta_x \otimes \ket{y^{(1)}}\bra{y^{(1)}} \otimes \ldots \otimes \ket{y^{(m)}}\bra{y^{(m)}} \right)\left(p_y \sigma_{R,y} \otimes \ket{y}\bra{y}^{\otimes m}\right)\right]\\
    &= \sum_{y \in \mcX^m} \sum_{\theta \in \Theta} p_{\Theta}(\theta) p_y\Tr\left[\mcM^\theta_{y_{f(\theta)}}\sigma_{R,y}\right]\\
    &= \sum_{y \in \mcX^m} \sum_{\theta \in \Theta} p_{\Theta}(\theta)\Tr\left[\mcM^\theta_{y_{f(\theta)}}\Tr_{P_1...P_m}\left[\left(\bbI_{R} \otimes \mcP^{f^{-1}(1)}_{1,y_1} \otimes \ldots \otimes \mcP^{f^{-1}(m)}_{m,y_m}\right)\rho_{RP_1...P_m}\right]\right]\\
    &= \sum_{y \in \mcX^m} \sum_{\theta \in \Theta} p_{\Theta}(\theta)\Tr\left[\left(\mcM^\theta_{y_{f(\theta)}} \otimes \mcP^{f^{-1}(1)}_{1,y_1} \otimes \ldots \otimes \mcP^{f^{-1}(m)}_{m,y_m}\right)\rho_{RP_1...P_m}\right]\\
    &= \sum_{y_{f(\theta)} \in \mcX} \sum_{\theta \in \Theta} p_{\Theta}(\theta)\Tr\left[\left(\mcM^\theta_{y_{f(\theta)}} \otimes \bbI_{P_1} \otimes \ldots  \otimes \mcP^{f^{(-1)}(f(\theta))}_{f(\theta),y_{f(\theta)}}\otimes \ldots \otimes \bbI_{P_m}\right)\rho_{RP_1...P_m}\right]\\
    &= \sum_{y_{f(\theta)} \in \mcX} \sum_{\theta \in \Theta} p_{\Theta}(\theta)\Tr\left[\left(\mcM^\theta_{y_{f(\theta)}} \otimes \mcP^{\theta}_{f(\theta),y_{f(\theta)}}\right)\rho_{RP_{f(\theta)}}\right]\\
    &= \sum_{x \in \mcX} \sum_{\theta \in \Theta} p_{\Theta}(\theta)\Tr\left[\left(\mcM^\theta_{x} \otimes \mcP^{\theta}_{f(\theta),x}\right)\rho_{RP_{f(\theta)}}\right]
\end{align*}
which is exactly the upper bound for $p_{\mcG_m}(\mcS)$. Therefore, by taking the supremum over all strategies for $\mcG_m$, we have $\omega_m(\mcR)\leq \omega_c(\mcR)$, which implies $\omega_m(\mcR)=\omega_c(\mcR)$.

\textbf{(5):} Let $\mcS$ be a classical-quantum strategy for $\mcG_m$, and without loss of generality assume that the first player is classical, i.e. $d_1=1$. Then, $\mcP^\theta_{1,x} = p^\theta_x \in \bbR$, where $p^\theta_x \geq 0$ and $\sum_{x \in \mcX}p^\theta_x=1$ for each $\theta \in \Theta$. Thus,
\begin{align*}
    p_{\mcG_m}(\mcS) &= \sum_{\theta \in \Theta}p_{\Theta}(\theta) \sum_{x \in \mcX} \Tr\left[\left(\mcM^\theta_x \otimes \mcP^\theta_{1,x} \otimes \ldots \otimes \mcP^\theta_{m,x}\right) \rho_{RP_1...P_m}\right]\\
    &\leq \sum_{\theta \in \Theta}p_{\Theta}(\theta) \sum_{x \in \mcX} \Tr\left[\left(\mcM^\theta_x \otimes \mcP^\theta_{1,x} \otimes \bbI_{P_2} \otimes \ldots \otimes \bbI_{P_m}\right) \rho_{RP_1...P_m}\right]\\
    &= \sum_{\theta \in \Theta}p_{\Theta}(\theta) \sum_{x \in \mcX} \Tr\left[\left(\mcM^\theta_x \otimes \mcP^\theta_{1,x}\right)\rho_{RP_1}\right]\\
    &= \sum_{\theta \in \Theta}p_{\Theta}(\theta)\sum_{x \in \mcX} \Tr\left[p^\theta_x\mcM^\theta_x\rho_R\right]\\
    &\leq \omega_c(\mcR)
\end{align*}
where the last inequality holds since the first player is classical, and shares no entanglement with the referee due to having a subsystem of dimension one.
\end{proof}

\gameValues*

\begin{proof}
    We will prove the claims in order.

    \textbf{(1):} Let $p,q$ be probability distributions on $[m]$ such that $p \preceq q$. A standard result of majorization theory (see, for instance, \cite{Majorization:Marshall79}) states that there exists a doubly stochastic matrix $\Pi \in \bbR^{m \times m}$ such that $p=\Pi q$ (where $p,q$ are interpreted as column vectors in $\bbR^m$). Furthermore, the Birkhoff-von Neumann theorem states that every doubly stochastic matrix $\Pi \in \bbR^{m \times m}$ can be expressed as a convex combination of permutation matrices $\Pi_{\sigma}$ for $\sigma \in S_m$.
    \[(\Pi_{\sigma})_{ij} = \begin{cases}
        1 \text{ if } i=\sigma(j)\\
        0 \text{ if } i \neq \sigma(j)
    \end{cases}\]
    So in general, we may write $\Pi = \sum_{\sigma \in S_m} r_{\sigma} \Pi_\sigma$ for some probability mass function $r_\sigma$ on $S_m$, and consequently $p_i = \sum_{j=1}^m \sum_{\sigma \in S_m } r_\sigma(\Pi_{\sigma})_{ij}q_j$.

    Now, let $\mcS_p$ be an arbitrary strategy for the OTG game $\mcG_p$ with referee configuration $\mcR$ and player selection probability $p$. We will assume without loss of generality that each player has a subsystem of dimension $D \in \bbN$, as any subsystem of dimension $d_i \leq D$ can be embedded into a $D$-dimensional Hilbert space faithfully. We will show that there exists a strategy $\mcS_q$ for $\mcG_q$ with referee configuration $\mcR$ and player selection probability $q$, that ``simulates" $\mcS_p$. First, we'll write $R_i$ to denote the winning probability of the $i$-th player, namely $R_i=\sum_{\theta \in \Theta}p_{\Theta}(\theta)\sum_{x \in \mcX}\Tr\left[\left(\mcM^\theta_x \otimes \mcP^\theta_{i,x}\right) \rho_{RP_i}\right]$, where the state $\rho_{RP_1...P_m}$ and the POVMs $\mcP^\theta_i$ come from the strategy $\mcS_p$. Then, the winning probability of strategy $\mcS_p$ in $\mcG_p$ can be concisely expressed as $p_{\mcG_p}(\mcS_p) = \sum_{i=1}^m p_iR_i$. Now, we will construct a randomized strategy $\mcS_q$ from $\mcS_p$, in which player $i$ simulates the guessing behavior of player $j$ with probability $\Pi_{ji}$. The new strategy $\mcS_q$ will use the classical-quantum state $\tau_{RP_1...P_m}$, where
    \[\tilde{\rho}_{RP_1...P_m} = \sum_{\sigma \in S_m} r_{\sigma} \ket{\sigma(1)}\bra{\sigma(1)}_{P_1} \otimes \ldots \otimes \ket{\sigma(m)}\bra{\sigma(m)}_{P_m} \otimes \rho_{RP_{\sigma(1)}...P_{\sigma(m)}} \in \mcD(C_{P_1...P_m} \otimes \mcH_{RP_1...P_m})\]
    and $C_{P_i}$ is a classical register with $m$ states. The player POVMs $\tilde{\mcP}^\theta_i$ are also from $\mcS_p$:
    \[\tilde{\mcP}^{\theta}_{i,x} = \sum_{j=1}^m \ket{j}\bra{j}_{P_i} \otimes \mcP^{\theta}_{j,x}\]
    Under the new strategy $\mcS_q$, we show that the winning probability in the OTG game $\mcG_q$ is the same as the winning probability of the strategy $\mcS_p$ in the OTG game $\mcG_p$.
    \begin{align*}
        p_{\mcG_q}(\mcS_q) &= \sum_{\theta \in \Theta}\sum_{i=1}^mp_{\Theta}(\theta)q_i \sum_{x \in \mcX} \Tr\left[\left(\mcM^\theta_{x} \otimes \tilde{\mcP^\theta_{i,x}}\right)\tilde{\rho}_{RP_i}\right]\\
        &= \sum_{\theta \in \Theta}\sum_{i,j=1}^m\sum_{\sigma \in S_m}p_{\Theta}(\theta)q_i r_\sigma\sum_{x \in \mcX} \Tr\left[\left(\ket{j}\bra{j} \otimes \mcM^\theta_{x} \otimes \mcP^\theta_{j,x}\right)\left(\ket{\sigma(i)}\bra{\sigma(i)} \otimes \rho_{RP_{\sigma(i)}}\right)\right]\\
        &= \sum_{\theta \in \Theta}\sum_{i=1}^m\sum_{\sigma \in S_m}p_{\Theta}(\theta)q_i r_\sigma\sum_{x \in \mcX} \Tr\left[\left(\mcM^\theta_{x} \otimes \mcP^\theta_{\sigma(i),x}\right) \rho_{RP_{\sigma(i)}}\right]\\
        &= \sum_{\sigma \in S_m}\sum_{i=1}^mq_i r_\sigma R_{\sigma(i)} = \sum_{\sigma \in S_m}\sum_{i=1}^mq_{\sigma^{-1}(i)} r_\sigma R_{i} = \sum_{\sigma \in S_m}\sum_{i=1}^m \left(\Pi_\sigma q\right)_i r_\sigma R_{i}\\
        &= \sum_{i=1}^m \left(\sum_{\sigma \in S_m} r_\sigma\Pi_\sigma q\right)_i R_{i} = \sum_{i=1}^m \left(\Pi q\right)_i R_{i} = \sum_{i=1}^m p_i R_{i} = p_{\mcG_p}(\mcS_p)
    \end{align*}
    Therefore, any arbitrary strategy $\mcS_p$ for the OTG game $\mcG_p$ can be simulated by a strategy $\mcS_q$ for the OTG game $\mcG_q$, achieving the same winning probability. Taking supremums, we have $\omega(\mcR,p) \leq \omega(\mcR,q)$. Since the strategy $\mcS_q$ only involves an overhead of classical randomness in relation to $\mcS_p$, the dimension of each subsystem in the strategy $\mcS_q$ is still bounded above by $D$. Hence, $\omega(\mcR,p,D) \leq \omega(\mcR,q,D)$.

    \textbf{(2):} The proof is exactly the same as the case for MoE games. Every $\underline{d}$-bounded strategy for $\mcG$ is automatically a $\underline{d}'$-bounded strategy for $\mcG$, and a $\underline{d}'$-bounded strategy for $\mcG$ is, trivially, a strategy for $\mcG$. Hence we arrive at the chain of inequalities $\omega(\mcR,q,\underline{d})\leq \omega(\mcR,q,\underline{d}')\leq \omega(\mcR,q)$ by taking supremums.

    \textbf{(3):} The claim actually follows from claim (1). Since $q_{\text{Unif}} = \Pi q_{\text{Unif}}$ with $\Pi=\frac{1}{m!}\sum_{\sigma \in S_m}\Pi_\sigma$, then starting from any strategy $\mcS$, we can perform the simulation procedure of (1) to obtain a strategy $\mcS'$ that achieves the same winning probability in the OTG game $\mcG$ as strategy $\mcS$. In the strategy $\mcS'$, the player POVMs are identical, and because $\Pi$ is a uniform combination of all permutation matrices, the shared state $\tilde{\rho}_{RP_1...P_m}$ is invariant under permuting the player indices. Thus, any strategy $\mcS$ can be simulated with a symmetric strategy $\mcS'$, achieving the same winning probability. Taking the supremum over all strategies, it follows that $\omega(\mcR,q)$ can be approximated with symmetric strategies. Additionally, since only an overhead of classical randomness is incurred in going from $\mcS$ to $\mcS'$, it follows that $\omega(\mcR,q,D)$ can be approximated with symmetric strategies. In fact, due to compactness, $\omega(\mcR,q,D)$ can be achieved exactly with symmetric strategies. 

    \textbf{(4):} The three Pauli-basis OTG game $\mcG$ of Section~\ref{sec:PauliOTG} gives an example of an OTG game in which the value, classical value, and bipartite value are all distinct, namely
    \[\omega_c(\mcR,q_{\text{Unif}}) = \frac{1}{2}+\frac{\sqrt{3}}{6}<\omega_{bpe}(\mcR,q_{\text{Unif}}) = \frac{1}{2}+\frac{\sqrt{6}}{8}<\omega(\mcR,q_{\text{Unif}})=\frac{5}{6}.\]
    Hence, the OTG game $\mcG$ has both a quantum and super bipartite advantage.
\end{proof}

\compareValues*

\begin{proof}
    The first inequality follows from Theorem~\ref{thm:moe_values}, whereas the second to last inequality follows from Theorem~\ref{thm:game_values}. The last inequality is trivial, and is equivalent to the statement that any strategy $\mcS$ wins an OTG game $\mcG$ with probability at most one. It remains for us to prove the second inequality and the middle equality.

    Let $\mcG_m$ be the $m$-player MoE game with referee configuration $\mcR$, and let $\mcG_q$ be the OTG game with referee configuration $\mcR$ and player selection distribution $q$ on $m$-players. We'll consider an arbitrary strategy $\mcS$, and observe the winning probability of $\mcS$ in each game.
    \begin{align*}
        p_{\mcG_m}(\mcS) &= \sum_{\theta \in \Theta}p_\Theta(\theta)\sum_{x \in \mcX} \Tr\left[\left(\mcM^\theta_x \otimes \mcP^\theta_{1,x} \otimes \ldots \otimes \mcP^\theta_{i,x} \otimes \ldots \otimes \mcP^\theta_{m,x}\right)\rho_{RP_1...P_m}\right]\\
        &\leq \sum_{\theta \in \Theta}p_\Theta(\theta)\sum_{x \in \mcX} \Tr\left[\left(\mcM^\theta_x \otimes \bbI_{P_1} \otimes \ldots \otimes \mcP^\theta_{i,x} \otimes \ldots \otimes \bbI_{P_m} \right)\rho_{RP_1...P_m}\right] ~~~ \text{for all } i \in [m]
    \end{align*}
    Since the inequality holds for all $i \in [m]$, it also holds for the index $i$ that minimizes the right-hand side of the second line, which is simply $p_{\mcR,\min}(\mcS)$. Consequently, $\omega_m(\mcP)\leq \omega_{m,\max\min}(\mcR)$. Now, we compare the minimum winning probability $p_{\mcR,\min}(\mcS)$ of the strategy $\mcS$ to the winning probability of the strategy $\mcS$ in the OTG game $\mcG_q$.
    \begin{align*}
        p_{\mcR,\min}(\mcS) &= \min_{i \in [m]}\sum_{\theta \in \Theta}p_\Theta(\theta)\sum_{x \in \mcX} \Tr\left[\left(\mcM^\theta_x \otimes \bbI_{P_1} \otimes \ldots \otimes \mcP^\theta_{i,x} \otimes \ldots \otimes \bbI_{P_m} \right)\rho_{RP_1...P_m}\right]\\
        &\leq \sum_{i=1}^m q_i\sum_{\theta \in \Theta}p_\Theta(\theta)\sum_{x \in \mcX} \Tr\left[\left(\mcM^\theta_x \otimes \bbI_{P_1} \otimes \ldots \otimes \mcP^\theta_{i,x} \otimes \ldots \otimes \bbI_{P_m} \right)\rho_{RP_1...P_m}\right]\\
        &= p_{\mcG_q}(\mcS)
    \end{align*}
    So the minimum winning probability os bounded above by the winning probability in the OTG game with player selection distribution $q$, hence $\omega_{m,\max\min}(\mcR)\leq \omega(\mcR,q)$. Choosing $q=q_{\text{Unif}}$ gives $\omega_{m,\max\min}(\mcR)\leq \omega(\mcR,q_{\text{Unif}})$. To prove that this inequality is actually an equality, we recall from Theorem~\ref{thm:game_values} that $\omega(\mcR,q_{\text{Unif}})$ can be approximated with symmetric strategies. Since all players have the same winning probability in a symmetric strategy, it follows that $p_{\mcR,\min}(\mcS)=p_{\mcG_q}(\mcS)$ for symmetric strategies $\mcS$, and taking the supremum over all symmetric strategies $\mcS$, we arrive at $\omega_{m,\max\min}(\mcR)=\omega(\mcR,q_{\text{Unif}})$.
\end{proof}

\end{document}